\RequirePackage{tikz}

\documentclass[sn-mathphys]{sn-jnl}

\usepackage{amsmath}
\usepackage{amssymb}
\usepackage{mathtools}

\usepackage{physics}
\usepackage{cleveref}

\usepackage{enumitem}
\usepackage{dsfont}

\usepackage{aligned-overset}

\usepackage{verbatim}

\usepackage{tabularray}

\usepackage{rotating}

\jyear{2021}%

\theoremstyle{thmstyleone}%
\newtheorem{theorem}{Theorem}
\newtheorem{proposition}[theorem]{Proposition}%
\newtheorem{lemma}[theorem]{Lemma} 
\newtheorem{corollary}[theorem]{Corollary} 
\newtheorem{remark}[theorem]{Remark}

\theoremstyle{thmstyletwo}%

\theoremstyle{thmstylethree}%

\theoremstyle{assumption}
\newtheorem{assumption}{Assumption}

\newcommand\sgn{\operatorname{sgn}}
\newcommand\RE{\operatorname{Re}}
\newcommand\IM{\operatorname{Im}}
\renewcommand\MOD{\operatorname{mod}}

\newcommand\bb[1]{\mathbb{#1}} 
\newcommand\q[1]{\mathcal{#1}} 

\newcommand{\Complex}{\mathbb{C}}
\renewcommand{\R}{\mathbb{R}}
\renewcommand{\Z}{\mathbb{Z}}
\newcommand{\N}{\mathbb{N}}

\renewcommand{\T}{\mathbb{T}}

\newcommand*{\defeq}{\mathrel{\vcenter{\baselineskip0.5ex \lineskiplimit0pt\hbox{\scriptsize.}\hbox{\scriptsize.}}} =}
\newcommand*{\defeqr}{=\mathrel{\vcenter{\baselineskip0.5ex \lineskiplimit0pt \hbox{\scriptsize.}\hbox{\scriptsize.}}} }

\definecolor{TUMBlau}{RGB}{0,101,189} 
\definecolor{TUMOrange}{RGB}{227,114,34} 
\definecolor{TUMGreen}{RGB}{162,173,0}

\def\BibTeX{{\rm B\kern-.05em{\sc i\kern-.025em b}\kern-.08em T\kern-.1667em\lower.7ex\hbox{E}\kern-.125emX}}

\begin{document}

\title[Background Denoising for Ptychography via Wigner Distribution Deconvolution]{Background Denoising for Ptychography via Wigner Distribution Deconvolution}


\author[1,2]{\fnm{Oleh} \sur{Melnyk} \email{oleh.melnyk@tu-berlin.de}}
\author*[1,3]{\fnm{Patricia} \sur{Römer} \email{roemerp@ma.tum.de} 
}

\affil[1]{\small  \orgdiv{Mathematical Imaging and Data Analysis}, \orgname{Helmholtz Center Munich}, \orgaddress{\street{Ingolstädter Landstrasse 1}, \city{Neuherberg}, \postcode{85764}, \country{Germany}, } 
}

\affil[2]{\small \orgdiv{Institute of Mathematics}, \orgname{Technical University of Berlin}, \orgaddress{\street{Straße des 17. Juni 136}, \city{Berlin}, \postcode{10623}, \country{Germany}}}

\affil[3]{\small \orgdiv{Department of Mathematics}, \orgname{Technical University of Munich}, \orgaddress{\street{Boltzmannstr. 3}, \city{Garching bei M\"unchen}, \postcode{85748}, \country{Germany}}} 



\abstract{
Ptychography is a computational imaging technique that aims to reconstruct the object of interest from a set of diffraction patterns. Each of these 
is obtained by a localized illumination of the object, which is shifted after each illumination to cover its whole domain. As in the resulting measurements the phase information is lost, ptychography gives rise to solving a phase retrieval problem.
In this work, we consider ptychographic measurements corrupted with background noise, a type of additive noise that is independent of the shift, i.e., it is the same for all diffraction patterns.
Two algorithms are provided, for arbitrary objects and for so-called phase objects that do not absorb the light but 
only scatter it. For the second type, a uniqueness of reconstruction is established for almost every object. Our approach is based on the Wigner Distribution Deconvolution, which lifts the object to a higher-dimensional matrix space where the recovery can be reformulated as a linear problem. Background noise only affects a few equations of the linear system that are therefore discarded. The lost information is then restored using redundancy in the higher-dimensional space.
}





\maketitle

\section{Introduction}

Ptychography \cite{hoppe1969beugung} is an imaging technique that allows recovery of an object from a collection of diffraction patterns. In a ptychographic experiment, a beam of light is concentrated on a small part of the object of interest. As light passes through the object, it encodes information about the object. Then, a detector placed in the far field captures the intensity of the incoming light wave. Subsequently, the object is shifted and the experiment with the localized illumination is performed again for the next region. The two adjacent illuminated areas of the object are required to overlap, effecting that multiple measurements contain information on the same part of the object. That means there is redundancy in the data, which ensures that the object can be recovered from the obtained measurements. Finally, the described procedure is repeated until the whole object is covered. 

Ptychography is applied, e.g., in X-ray microscopy \cite{pfeiffer2018x} as well as in electron microscopy \cite{rodenburg2019ptychography}, to obtain high resolution images at nanometer scale of biological specimen \cite{giewekemeyer2010quantitative} or other nanoscale materials \cite{shi2019x}.

The illumination is realized as a localized window function $g$ and the object is described by its object transmission function $f$. The transmission function represents two physical properties of the object. Its amplitude quantifies the percentage of the illumination that is not absorbed by the object. A zero value means that in the respective part of the object no light passes through, while in the other extreme case the value is set to one. On the other hand, the phase of the object transmission function represents the scattering of the illumination by the object. Often in practice, objects are sufficiently thin and only scatter light. Consequently, their transmission function has constant one magnitude and they are called phase objects \cite{hawkes2019springer}. 

Mathematically, the obtained diffraction patterns are given as follows. The light exiting the object is characterized by the product of the illumination with the object transmission function. When the exit wave propagates to the far field, it can be described as the Fourier transform
\begin{equation*}
    \int_{\R^2} f(y) g(y-s) e^{-2\pi i \xi \cdot y} dy,
\end{equation*}
with frequency variable $\xi$, and $s$ denoting the shift. Due to the nature of charge-coupled device (CCD) cameras used as detectors, the observed images are the intensities of the incoming waves, which are given by
\begin{equation*}
    I(s,\xi) = \left\vert \int_{\R^2} f(y) g(y-s) e^{-2\pi i \xi \cdot y} dy \right\vert^2.
\end{equation*}

The task is then to recover the object transfer function from samples of $I(s,\xi)$. Reconstruction from intensity measurements means to solve a phase retrieval problem. In this work, we study a discrete version of the ptychographic inverse problem. The object function $f$ is approximated by a vector $x \in \Complex^d$. The measurement procedure involving the window function $g$, the translation in the space variable, and the modulation with the complex exponential are summarized in their discrete correspondence as masks $m^{(r,\ell)} \in \Complex^d$, where $r$ represents the respective shift and $\ell$ the frequency. Together, the discrete version of the ptychographic measurements reads as
\begin{equation*}
    Y_{r,\ell} = \left\vert \left\langle x, m^{(r,\ell)} \right\rangle \right\vert^2.
\end{equation*}

For computational reconstruction, a phase retrieval problem is often posed as an optimization problem, which is then solved by an iterative method. The pool of algorithms includes (stochastic) gradient methods \cite{candes2015phase, xu2018accelerated, wang2017solving, melnyk2022stochastic}, alternating projection and reflection methods \cite{gerchberg1972practical, fienup1978reconstruction, marchesini2016alternating, luke2004relaxed}, and the alternating direction method of multipliers \cite{chang2018total, chang2019blind}. Among practitioners, the most popular algorithm is a version of stochastic gradient descent for ptychography, the so-called ptychographic iterative engine \cite{rodenburg2004phase, melnyk2023convergence}. 

Alternatively, the ptychographic inverse problem can also be tackled with a non-iterative solver based on the Wigner Distribution Deconvolution (WDD) method proposed in \cite{rodenburg1992theory, chapman1996phase}. By this approach, the measurements are transformed into a convolution of object- and illumination-related functions, which are then decoupled by a deconvolution procedure. The remaining step is to recover the object from the corresponding function. The WDD method is described in more detail in \Cref{sec: ptycho and wdd}.

Contributions in \cite{iwen2016fast, iwen2020phase, preskitt2018phase, cordor2020fast, perlmutter2020provably, perlmutter2021inverting, melnyk2023phase} use the WDD algorithm to prove uniqueness of reconstruction for the ptychographic problem. Furthermore, these works have shown that the WDD algorithm is robust to additive noise. Note that uniqueness and stability are the major differences between ptychography \cite{bojarovska2016phase, jaganathan2016stft, alaifari2021stability, bendory2022nearoptimal} and a single illumination Fourier phase retrieval problem \cite{beinert2015ambiguities}.
The data redundancy generated by the overlapping illuminations allows to solve the phase retrieval problem uniquely.

In a ptychographic experiment under real-world conditions, different issues occur, one of which is experimental noise. In an imaging experiment, one type of noise arises from the counting procedure of the illumination particles arriving at the detector, causing the measurements to be corrupted by random Poisson noise \cite{thibault2012maximum, roemer2022wirtinger, li2022poisson}. Further, CCD cameras face the problem that imperfections in the experimental setup such as, e.g., contamination, result in additional charge measured in the pixels of the detector. This type of noise is referred to as background noise or parasitic scattering \cite{chang2019advanced, salditt2020nanoscale}. While Poisson noise can be neglected in an experimental setup using a sufficiently high illumination particle count, background noise can highly dominate the ptychographic measurements. For this reason, we focus here on a noise model that assumes background noise to be the only perturbation source.

Background noise is assumed to be independent of the measurement process. That is, for each shift of the object, the recorded diffraction pattern is composed of the actual captured intensity and a shift-independent background noise,
\begin{equation*}
    \left\vert \int_{\R^2} f(y) g(y-s) e^{-2\pi i \xi \cdot y} dy \right\vert^2 + b(\xi).
\end{equation*}
In the discrete version of the problem, the resulting measurements are of the form
\begin{equation*}
    Y_{r,\ell} = \left\vert \left\langle x, m^{(r,\ell)} \right\rangle \right\vert^2 + b_\ell, 
\end{equation*}
with background noise $b \in \R^d$.

When approaching the ptychographic phase retrieval problem as an optimization problem, the background can be incorporated into the optimization procedure as an additional unknown. Suitable iterative methods for tackling the resulting problem were presented in \cite{marchesini2013augmented, chang2019advanced}. Another possibility is to preprocess the measurements to counteract the effects of the background noise \cite{wang2017background}.

In this paper, we make use of the WDD method to develop an algorithm that removes the background noise to full extent. Our approach uses the shift-invariance of the noise, which allows to separate the background from most of the noiseless intensities in the deconvolution step. Consequently, the remaining corrupted intensities can be discarded to denoise the data completely. We find that their noise-free equivalent can be recovered from the separated noise-free information using the redundancy in the ptychographic measurements. After that, we can continue with the object reconstruction as in the WDD algorithm. Two different recovery procedures are proposed for phase and arbitrary objects. Our main contribution is a guarantee for uniqueness of reconstruction from ptychographic measurements with background noise for almost every phase object.

The paper is structured in the following way. In \Cref{sec: preliminaries} we provide the necessary notation and basic properties related to the Fourier transform. We formulate the ptychographic problem mathematically and summarize preliminary results for the WDD approach. Our contribution is presented in \Cref{sec: background_noise_removal} and proved in \Cref{sec: proofs}. In \Cref{sec: numerics}, we corroborate the theoretical findings with numerical experiments. Finally, the paper is summarized by a short conclusion.

\section{Preliminaries}\label{sec: preliminaries}

\subsection{Notation and basic properties}

In this paper, we work with index sets $[d] \defeq \left\{0,\ldots,d-1\right\}$ for $d \in \N$ and the entries of vectors $x \in \bb C^d$ are enumerated from $x_0$ to $x_{d-1}$. Note that all indices are considered modulo $d$, but we forgo the notation using $\MOD d$.
The Euclidean norm of a vector $x \in \Complex^d$ is denoted by $\norm{x}_2$, and the Frobenius norm of a matrix $A \in \Complex^{d_1 \times d_2}$, $d_1,d_2 \in \N$ by $\norm{A}_F$. 

The notation $a \in b + c\Z$ means that $a$ is equal to $b$ up to an additive constant from $c\Z \defeq \left\{cm: m \in \Z\right\}$ with $c \in \R$.

Any $z\in\Complex$ is composed as $z = \vert z \vert e^{i\phi}$ with $\vert z \vert \in \R_{\geq 0}$ and $\phi \in (0,2\pi]$. We call $\vert z \vert$ the magnitude, and $\phi \defeqr \arg(z)$ the argument of $z$. For $z \in \Complex \backslash \left\{0\right\}$, $ \mathrm{sgn}(z) \defeq \frac{z}{\vert z \vert}$ denotes the phase of $z$. For $z = 0$, it is set $\mathrm{sgn}(0) \defeq 0$. For $x \in \Complex^d$, the operations $\vert x \vert$, $\arg(x)$, $\sgn(x)$ are applied entrywise.

We denote the discrete Fourier transform of a vector $x \in \Complex^d$, and its inverse, by
\begin{equation*}
    (Fx)_j \defeq \sum_{k \in [d]} e^{-\frac{2\pi i kj}{d}} x_k, \quad \text{and} \quad (F^{-1}x)_j \defeq \frac{1}{d} \sum_{k \in [d]} e^{\frac{2\pi i kj}{d}} x_k, \quad j \in [d],
\end{equation*}
involving the Fourier matrix $F \in \Complex^{d \times d}$ with entries $(F)_{k,j} = e^{-\frac{2\pi i kj}{d}}$.

The circular shift, modulation and reflection operators on $\Complex^{d}$ with parameter $r \in \Z$ are defined as
\begin{equation*}
(S_{r}x)_j \defeq x_{j+r }, 
\quad 
(M_{r}x)_j \defeq e^{\frac{2\pi i jr}{d}}x_{j},
\quad 
(R_dx)_j \defeq x_{-j},
\quad j \in [d], 
\end{equation*}
respectively. We will use the following basic relations between these operators and  the Fourier transform.
\begin{lemma} \label{l: basic_properties}
For all $x \in \Complex^d$ and $r \in [d]$, we have
\[
(i) \ F (S_r x) = M_r (F x),
\quad \text{and} \quad
(ii) \ F \overline x = R_d \overline{F x}.
\]
\end{lemma}

For real vectors, \Cref{l: basic_properties} $(ii)$ turns into the following well-known result.
\begin{corollary}
\label{l: fourier of real vector conjugate symmetric}
The Fourier transform of $x$ is conjugate symmetric, i.e., it satisfies $Fx = R_d \overline{Fx}$, if and only if $x \in \R^d$.
\end{corollary}

For two vectors $x,y \in \Complex^d$, the Hadamard product $x \circ y$ is defined by the entrywise products $(x\circ y)_j \defeq x_jy_j, \ j \in [d]$. The discrete circular convolution of $x$ and $y$ is defined as 
\begin{equation*}
(x * y)_j \defeq \sum_{k\in [d]} x_{j-k} ~ y_k, \quad j\in [d].
\end{equation*}
The relation 
\begin{equation}\label{thm: convolution theorem}
dF(x\circ y) = Fx * Fy, \quad x,y \in \Complex^d,
\end{equation}
is known as the (discrete) convolution theorem. A further well-known relation we need is Plancherel's identity
\begin{equation} \label{thm: plancherel}
\norm{Fx}_2^2 = d\norm{x}_2^2, \quad x \in \Complex^d.
\end{equation}

\subsection{Ptychography and Wigner Distribution Deconvolution}
\label{sec: ptycho and wdd}

In ptychography, we consider short-time Fourier transform measurements
\begin{equation}\label{eq: ptycho_measurements_noise}
     \tilde{Y}_{\ell,r} = \left\vert \left(F \left[S_{-r}x \circ w\right]\right)_\ell\right\vert^2 + N_{\ell,r}
     =  \bigg\vert \sum_{k \in [d]} e^{-\frac{2\pi i k \ell}{d}} x_{k-r} w_k \bigg\vert^2  + N_{\ell,r},
\end{equation}
where $x \in \Complex^d$ is the object of interest, and $N = \left(N_{\ell,r} \right)_{\ell \in [d], r \in [d]} \in \R^{d\times d}$ denotes noise. Furthermore, $w \in \Complex^d$ represents, in terms of physics nomenclature, the illumination. In the mathematical community, this is more commonly referred to as window. The window $w$ is assumed to be localized, i.e., $\mathrm{supp}(w) = [\delta]$ for $\delta < d$. Throughout the paper, we assume that the window $w$ is known, hence the problem lies in recovering the object $x$ from the measurements \eqref{eq: ptycho_measurements_noise}. That means, the ptychographic problem is an inverse problem. More specifically, it can be understood as a phase retrieval problem with short-time Fourier transform measurements.

Defining $\T \defeq \left\{ \beta \in \Complex: \vert \beta \vert = 1\right\}$, we obtain the same measurements 
\begin{equation*}
\left\vert \left(F \left[S_{-r}(\alpha x) \circ w\right]\right)_\ell\right\vert^2 =  \left\vert \left(F \left[S_{-r}x \circ w\right]\right)_\ell\right\vert^2     
\end{equation*}
for all $\alpha \in \T$, meaning that phaseless measurements only allow unique recovery up to a global phase. This motivates to define an equivalence relation
\begin{equation*}
x \sim \hat{x} \Leftrightarrow x = \alpha \hat{x} ~\text{for some}~\alpha \in \T,
\end{equation*}
and to call a solution to a phase retrieval problem unique if it is an element of
\begin{equation*}
    \left\{\alpha x: \alpha \in \T\right\}.
\end{equation*}

Uniqueness and stability of the ptychographic problem, in a noise-free setting, were investigated, e.g., in \cite{bojarovska2016phase, alaifari2021stability, bendory2017non}. These results are based on the idea of the Wigner Distribution Deconvolution (WDD) \cite{chapman1996phase, rodenburg2008ptychography}, which relates the ptychographic measurements to the Wigner distribution function of both the object and the window. 

\begin{theorem}[{\cite[Lemma 7]{perlmutter2021inverting}}] \label{thm: WDD_theorem}

    Let $Y \in \R^{d\times d}$ be the matrix with entries
    \begin{equation} \label{eq: ptychographic_measurements_noisefree}
     Y_{\ell,r} = \left\vert \left( F \left[S_{-r}x \circ w\right] \right)_\ell \right\vert^2, \quad \ell,r \in [d].
\end{equation}
Then, for all $j,k \in [d]$,
 \begin{equation} \label{eq: WDD} 
    \left( F^{-1} Y F \right) _{j,k}  = \left( F \left[x \circ S_j\overline{x}  \right] \right)_k \cdot  \overline{\left( F \left[\overline{w} \circ S_{j}w  \right] \right)}_k.
\end{equation}

\end{theorem}

To retrieve the Fourier coefficients $\left( F \left[x \circ S_j\overline{x}  \right] \right)_k$ via the relation (\ref{eq: WDD}), we need to assume the following.

\begin{assumption} \label{assumption: window}
   Let the window $w$ satisfy
\begin{equation*}
    \left( F \left[\overline{w} \circ S_{j}w  \right] \right)_k \neq 0 \quad \text{for all} ~ -\gamma < j < \gamma, ~ k \in [d],
\end{equation*} 
for some $0 < \gamma \leq \delta.$
\end{assumption}

With this assumption, the Wigner Distribution Deconvolution approach allows the following statement on uniqueness of reconstruction from ptychographic measurements.

\begin{theorem}[{\cite[Theorem 2.2]{bojarovska2016phase}}] \label{thm: bojarovska_flinth_general}
Let  $\delta > \frac{d}{2}$. If \Cref{assumption: window} holds with $\gamma = \delta$, any $x \in \Complex^d$ is uniquely determined by measurements \eqref{eq: ptychographic_measurements_noisefree}.
\end{theorem}

For non-vanishing objects, that is $x \in \Complex^d$ with $\min_{\ell \in [d]} \vert x_\ell\vert > 0$, the assumption on the window that guarantees uniqueness of reconstruction can be mitigated as follows.

\begin{theorem}[{\cite[Theorem 2.4]{bojarovska2016phase}}] \label{thm: bojarovska_flinth}
Let \Cref{assumption: window} hold with $\gamma \geq 2$. Then, any $x \in \Complex^d$ satisfying $\min_{\ell \in [d]} \vert x_\ell\vert > 0$ is uniquely determined by measurements \eqref{eq: ptychographic_measurements_noisefree}.
\end{theorem}

With \Cref{assumption: window}, the result in \Cref{thm: WDD_theorem} can be turned into a recovery formula.

\begin{corollary} \label{cor: WDD_recovery}
    If \Cref{assumption: window} holds for some $0<\gamma\leq \delta$, the Fourier coefficients $\left( F \left[x \circ S_j\overline{x}  \right] \right)_k$ for all $- \gamma < j < \gamma,$ and all $k \in [d]$, can be recovered by
    \begin{equation*}
        \left( F \left[x \circ S_j\overline{x}  \right] \right)_k =   \left( F^{-1} Y F \right) _{j,k} \ \big/ \  \overline{\left( F \left[\overline{w} \circ S_{j}w  \right] \right)}_k.
    \end{equation*}
\end{corollary}

Using the Fourier coefficients obtained by \Cref{cor: WDD_recovery}, the products $(x \circ S_j \overline{x})_\ell, \ \ell\in [d],$ can be reconstructed for all $- \gamma < j < \gamma$.  These vectors correspond to the diagonals of the rank-one matrix $x x^*$, which is why we will refer to the vectors $x \circ S_j \overline{x}$ as diagonals. 

As $\gamma \leq \delta < d$, only a part of the matrix $xx^*$ is recoverable. More precisely, if \Cref{assumption: window} is true for $\gamma \leq \delta$, we can reconstruct the matrix
\begin{align} \label{eq: matrix of diagonals}
    X_{\ell_1,\ell_2} \defeq \begin{cases}
        (x \circ S_{\ell_2 - \ell_1} \overline{x})_{\ell_1} = x_{\ell_1}\overline{x}_{\ell_2}, &\text{if} ~ \min\left\{ \vert \ell_1 - \ell_2 \vert, d - \vert \ell_1 - \ell_2\vert \right\} < \gamma,\\ 0, &\text{otherwise.}
    \end{cases} 
\end{align}
This matrix is symmetric, that means the lower off-diagonals provide the same information as the upper off-diagonals. Hence, in the following, we will avoid discussing the cases $- \gamma < j <0$.

Using the main diagonal of $X$, i.e., $x \circ  S_0\overline{x} = \vert x \vert^2$, the magnitudes of $x$ can be recovered. More stable approaches for magnitude estimation can be found in \cite{preskitt2021admissible, melnyk2023phase}. 

The first off-diagonal of $X$, i.e., $x \circ  S_1\overline{ x}$, provides the argument differences $\arg(x_\ell) - \arg(x_{\ell+1})$ for all $\ell \in [d]$. As $x$ can be recovered only up to its global phase, $\arg(x_0)$ can be chosen arbitrarily. The phases of $x_\ell, ~ \ell >0$, are then found recursively from the argument differences. This approach can be linked to the greedy angular synchronization method discussed in \cite{iwen2016fast}. Later, \cite{viswanathan2015fast, iwen2020phase} suggested to alternatively use an eigenvector-based approach to angular synchronization. Define the matrix of phase differences
\begin{align*}
\left(\mathrm{sgn}(X)\right)_{\ell_1,\ell_2} \defeq \begin{cases}
        \mathrm{sgn}(x_{\ell_1})\mathrm{sgn}(\overline{x}_{\ell_2}), &\quad X_{\ell_1,\ell_2} \neq 0,\\
        0, &\quad  \text{otherwise,}
    \end{cases}
\end{align*}
corresponding to $X$. Then, the vector of phases of $x$, i.e., $\mathrm{sgn}(x)$, is the top eigenvector of the matrix $\mathrm{sgn}(X)$. Other versions for phase synchronization, together with reconstruction guarantees, were presented in \cite{preskitt2018phase, filbir2021recovery}.

The main steps of the reconstruction algorithm for ptychographic data based on Wigner Distribution Deconvolution are summarized in \Cref{alg: algorithm_WDD}.

\RestyleAlgo{ruled}
\begin{algorithm} 
\caption{Recovery from ptychographic measurements}\label{alg: algorithm_WDD}
\vspace{1mm}
\begin{flushleft}

\textbf{Input:} Ptychographic measurements $Y \in \R^{d\times d}$ as in  \eqref{eq: ptychographic_measurements_noisefree} with $w \in \Complex^d$ satisfying $\mathrm{supp}(w) = [\delta]$ for $\delta < d$ and \Cref{assumption: window} for $0 <\gamma \leq \delta$.

\vspace{1mm}

\textbf{Step 1:} Compute $\left( F \left[x \circ S_j\overline{x}  \right] \right)_k$ for all $j \in [\gamma] , ~k \in [d],$ via \Cref{cor: WDD_recovery}.

\vspace{1mm}

\textbf{Step 2:} Via inverse Fourier transforms, compute $\left( x \circ S_j\overline{x} \right)_\ell$ for all $j \in [\gamma] , ~\ell \in [d],$ and build $X$ as in \Cref{eq: matrix of diagonals}.
	
\vspace{1mm}
 
\textbf{Step 3:} Compute the magnitudes $\vert \tilde{x}_\ell\vert = \sqrt{X_{\ell,\ell}},~ \ell \in [d]$.

\vspace{1mm}

\textbf{Step 4:} Compute the top eigenvector $\tilde{z}$ of  $\mathrm{sgn}(X)$. Set $\mathrm{sgn}(\tilde{x}) = \mathrm{sgn}(\tilde{z})$.

\vspace{1mm}

\textbf{Output:} $\tilde{x} = \vert \tilde{x} \vert \cdot \mathrm{sgn}(\tilde{x}) \in \Complex^d$ with $\tilde{x} \sim x$.
\end{flushleft}
\end{algorithm}

\Cref{alg: algorithm_WDD} was shown to satisfy the following recovery guarantee.

\begin{theorem}[{\cite[Theorem 1]{preskitt2018phase}}] \label{thm: recovery guarantee wdd}

Let $\delta > 2$ and $d \geq 4\delta$. Applied to noisy ptychographic measurements \eqref{eq: ptycho_measurements_noise}, \Cref{alg: algorithm_WDD} creates an estimate $\tilde{x} \in \Complex^d$ satisfying
\begin{align*}
    \min_{\theta \in (0,2\pi]} \norm{x - e^{i\theta} \tilde{x}}_2 \leq ~ & 24 \frac{\norm{x}_\infty}{\min_{\ell \in [d]} \vert x_\ell \vert^2} \cdot  \frac{d^{\tfrac{3}{2}}}{\delta^{\tfrac{5}{2}}} \cdot \frac{ \norm{N}_F}{\min_{j \in [\delta],k\in [d]} \vert F[\overline{w}\circ S_j w]_k\vert}\\[3pt] 
    &+  \sqrt{\frac{ \norm{N}_F}{\min_{j \in [\delta],k\in [d]} \vert F[\overline{w}\circ S_j w]_k\vert}} .
\end{align*}
\end{theorem}

\Cref{thm: recovery guarantee wdd} states that exact recovery of the ground-truth object from noise-free measurements is possible via \Cref{alg: algorithm_WDD}, i.e., the output $\tilde{x}$ of \Cref{alg: algorithm_WDD} indeed satisfies $\tilde{x} \sim x$. Moreover, it tells that noise can affect the reconstruction quality only up to some amount that depends on the level of noise. That means, it can be reasonable to apply the algorithm to noisy data. However, exact recovery of the ground-truth is not to be expected.

\section{Background noise removal} \label{sec: background_noise_removal}

In the following, we investigate an adaption of the WDD reconstruction method for a special type of noise, so-called background noise. 

We consider ptychographic measurements
\begin{equation}\label{eq: ptycho_measurements}
     \tilde{Y}_{\ell,r} = Y_{\ell,r} + b_{\ell}, \quad \ell,r \in [d],
\end{equation}
with unknown background noise $b_\ell \in \R, \ \ell \in [d]$, and noise-free measurements $Y_{\ell,r},  \ \ell,r \in [d]$, as in \eqref{eq: ptychographic_measurements_noisefree}. For each shift $r \in [d]$ of the object, the background is assumed to be constant, i.e., this type of noise is only frequency-dependent.

Applying WDD to measurements with background noise, \Cref{thm: recovery guarantee wdd} provides an upper bound on how much the resulting reconstruction can be offset from the ground-truth object by the noise. However, \Cref{alg: algorithm_WDD} makes no use of the prior knowledge about the noise structure. The background noise is the same for all diffraction patterns. Hence, there is a redundancy that can be incorporated into WDD to improve its performance. 

Firstly, we note that only the zeroth Fourier coefficients are affected by the background noise in the WDD method.

\begin{proposition} \label{prop: theorem_WDD_background}
For all $k, j \in [d]$, ptychographic measurements \eqref{eq: ptycho_measurements} satisfy
\begin{align*}
\big(F^{-1} \tilde{Y} F\big)_{j,k}
= \big(F\left[x \circ  S_j\overline{ x}\right]\big)_k \cdot \big(\overline{F\left[\overline{w} \circ  S_j w\right]}\big)_k + d \big(F^{-1} b \big)_j \mathds{1}_{k = 0} .
\end{align*}

\end{proposition}

\begin{proof}[Proof.]
Since the background noise is only frequency dependent, the noise-free and the noisy measurements in \eqref{eq: ptycho_measurements} are related via
\begin{equation*}
    \tilde{Y} = Y + \begin{pmatrix}
        \vert & & \vert\\
        b & \cdots & b\\
        \vert & & \vert
    \end{pmatrix},
\end{equation*}
with $b = (b_j)_{j \in [d]}$.
We apply the same transforms as in \Cref{thm: WDD_theorem} to the measurements $\tilde{Y}$ and obtain the following relation:
\begin{align*}
    F^{-1}\tilde{Y} F &= F^{-1} Y F + F^{-1} \begin{pmatrix}
        b & \cdots & b\\
    \end{pmatrix} F \\
    & = F^{-1} Y F + F^{-1} \begin{pmatrix}
        d b & 0 & \ldots & 0\\
    \end{pmatrix}
    = F^{-1} Y F + d\begin{pmatrix}
        F^{-1}b & 0 & \ldots & 0\\
    \end{pmatrix}.
\end{align*}
\end{proof}

If \Cref{assumption: window} holds for some $\gamma \leq \delta$, all Fourier coefficients $\left(F\left[x \circ  S_j\overline{ x}\right]\right)_k$ can be reconstructed exactly for $k > 0$ and all $j \in [\gamma]$ as in \Cref{cor: WDD_recovery}.

We choose to discard the components $ \left(F\left[x \circ  S_j\overline{ x}\right]\right)_0, \ j \in [d]$, which are corrupted by the background noise. Our strategy is to reconstruct the zeroth coefficients before proceeding with the rest of the steps in \Cref{alg: algorithm_WDD}.

\subsection{Reconstruction algorithm} \label{sec: reconstruction algorithm}

To reconstruct the zero frequencies, we use the higher-order relation between the diagonals induced by the rank-one structure of $x x^*$. That is, for all $j, \ell \in [d]$, the corresponding two diagonals are related by 
    \begin{align} \label{eq: equality_of_diagonals}
        (x \circ S_j \overline{x}) \circ S_{\ell}\overline{(x \circ S_j \overline{x})} =  (x \circ S_{\ell} \overline{x}) \circ S_{j}\overline{(x \circ S_{\ell} \overline{x})}.
    \end{align}
This equality was previously made use of to circumvent the cases where \Cref{assumption: window} does not hold, which led to the so-called subspace completion strategy \cite{forstner2020well}.

The relation \eqref{eq: equality_of_diagonals} can be expressed in terms of frequencies via the Fourier transform. The convolution theorem \eqref{thm: convolution theorem}, together with the  properties of shift, modulation, and reversal operators in \Cref{l: basic_properties}, yields
\[
d F[ (x \circ S_j \overline{x}) \circ S_{\ell}\overline{(x \circ S_j \overline{x})}]
= F[ x \circ S_j \overline{x}] * F[S_{\ell}\overline{(x \circ S_j \overline{x})}]
= f^j * M_\ell R_d \overline{f^j}.
\]
Here and in the following, we abbreviate the Fourier coefficients as
\begin{equation}\label{eq: def f}
     f_{k}^j \defeq \left(F\left[x \circ  S_j\overline{ x}\right]\right)_k
\end{equation}
for all $j, k \in [d]$.
Expanding the convolution gives
\begin{align} \label{eq: equality_fourier_diagonals}
\sum_{k = 0}^{d-1} e^{-\frac{2\pi i \ell (k-s)}{d}} f^j_k ~ \overline{f^j}_{k-s} 
= \sum_{k = 0}^{d-1} e^{-\frac{2\pi i j (k-s)}{d}} f^\ell_k ~ \overline{f^\ell}_{k-s}
\end{align}
for all $s \in [d]$.

If \Cref{assumption: window} is satisfied for some $\gamma \leq \delta$ and $\ell, j \in [\gamma]$, the summands are fully known for all $k \in [d]\backslash \left\{0,s\right\}$. Hence, we can sort equation (\ref{eq: equality_fourier_diagonals}) into unknown and known summands, and obtain the system of linear equations
\begin{align} \label{eq: subspace_completion_linear_system}
& e^{\frac{2\pi i j s}{d}}f^\ell_0 ~ \overline{f^\ell}_{-s} 
+ f^\ell_s ~ \overline{f^\ell}_{0} 
- e^{\frac{2\pi i \ell s}{d}}f^{j}_0 ~ \overline{f^j}_{-s} 
- f^j_s ~ \overline{f^j}_{0} \\ \notag
& \qquad = \sum_{\substack{k = 1 \\ k \neq s}}^{d-1} e^{-\frac{2\pi i \ell (k-s)}{d}} f^j_k ~ \overline{f^j}_{k-s} 
- e^{-\frac{2\pi i j (k-s)}{d}}f^\ell_k ~ \overline{f^\ell}_{k-s} \defeqr c_{\ell, j, s},
\end{align}
for all $s \in [d]$. Note that the case $s=0$ leads only to squared magnitudes $\vert f^j_0 \vert^2$, which is why it is left out.

This linear system has four real unknowns and can be solved to obtain the zero frequencies $f^j_0 = \left(F[x \circ S_j \overline{x}]\right)_0$ for all $j \in [\gamma]$.

The number of unknowns can be further reduced by considering the case $\ell=0$. Recall that $(x \circ S_0 \overline{x})_k = \vert x_k \vert^2 \in \R$ for all $k\in [d]$ so that $f_s^0 = \overline{f^0}_{-s}$, $s \in [d]$, and 
\[
f_0^0 = \sum_{k=0}^{d-1} (x \circ S_0 \overline{x})_k \ge 0.
\]
This can be incorporated into \eqref{eq: subspace_completion_linear_system}, giving
\[
\left( e^{\frac{2\pi i j s}{d}} + 1\right) f^0_s ~ f^0_0 
- f^{j}_0 ~ \overline{f^j}_{-s} - f^j_s ~ \overline{f^j}_{0}
= c_{0,j,s}, \quad s \in [d]\backslash\{0\}.
\]
If $z_{j,s} \defeq \left( e^{\frac{2\pi i j s}{d}} + 1\right) f^0_s$ is zero, the above equation becomes
\[
f^{j}_0 ~ \overline{f^j}_{-s} + f^j_s ~ \overline{f^j}_{0}
= - c_{0,j,s},
\]
or, equivalently,
\begin{equation}\label{eq: subspace 1}
\begin{aligned}
\RE f^{j}_0 ~ [\RE f^j_{-s} + \RE f^j_{s}]  + \IM f^{j}_0 ~ [\IM f^j_{-s} + \IM f^j_{s}] & = - \RE c_{0,j,s}, \\
\RE f^{j}_0 ~ [\IM f^j_{s} - \IM f^j_{-s}]  + \IM f^{j}_0 ~ [\RE f^j_{-s} - \RE f^j_{s}] & = - \IM c_{0,j,s}. \\
\end{aligned}
\end{equation}
Otherwise, by multiplying it with $\overline{z}_{j,s}$, we get
\begin{equation*}
f^0_0 \vert z_{j,s}\vert^2
- f^{j}_0 ~ \overline{f^j}_{-s} \overline{z}_{j,s} - f^j_s ~ \overline{f^j}_{0} \overline{z}_{j,s}
= c_{0,j,s} \overline{z}_{j,s}.
\end{equation*}
As $f^0_0 \vert z_{j,s}\vert^2 \in \R$, the imaginary part of the equation reads as
\begin{equation}\label{eq: subspace 2}
\begin{aligned}
& \RE f^{j}_0 ~ [\IM(f^j_{-s} z_{j,s}) - \IM(f^j_{s} \overline{z}_{j,s}) ]  + \IM f^{j}_0 ~ [\RE(f^j_{s} \overline{z}_{j,s}) - \RE(f^j_{-s} z_{j,s})] \\
& \qquad \qquad \qquad \qquad \qquad \qquad \qquad \qquad \qquad \qquad \qquad=  \IM(c_{0,j,s} \overline{z}_{j,s}).
\end{aligned}
\end{equation}
Hence, finding four real unknowns from \eqref{eq: subspace_completion_linear_system} can be reduced to a linear system for just two real unknowns formed by \eqref{eq: subspace 1} or \eqref{eq: subspace 2} for $s \in [d]\backslash\{0\}$. 

The last step is to reconstruct $f^0_0$ from \eqref{eq: subspace_completion_linear_system} using the case $s = 0$. We get
\[
\vert f^0_0 \vert^2 - \vert f^j_0 \vert^2 = c_{0,j,0},
\]
which, with $f^0_0 \geq 0$, yields 
\[
f_0^0 
= \sqrt{\vert f^0_0 \vert^2}
= c_{0,j,0} + \vert f^j_0 \vert^2,
\quad \text{for any} \ j \in [\gamma] \backslash \{0\}.
\]
For a better stability, we average over $j$,
\begin{equation} \label{eq: subspace 0 component}
f_0^0 = \left( \frac{1}{\gamma-1} \sum_{j=1}^{\gamma-1} \left[ c_{0,j,0} + \vert f^j_0 \vert^2 \right] \right)^{1/2}.    
\end{equation}

The steps above provide a recovery strategy for the lost frequencies, summarized in \Cref{alg: algorithm general object}. The theoretical analysis of this procedure is a complicated task which remains to be tackled in future work.

\RestyleAlgo{ruled}
\begin{algorithm}[h!] 
\caption{Recovery from ptychographic measurements with background noise}\label{alg: algorithm general object}
\vspace{1mm}
\begin{flushleft}

\textbf{Input:} Ptychographic measurements $\tilde{Y} \in \R^{d\times d}$ as in  \eqref{eq: ptycho_measurements}  with $w \in \Complex^d$ satisfying $\mathrm{supp}(w) = [\delta]$ for $\delta < d$ and \Cref{assumption: window} for $0 <\gamma \leq \delta$.

\vspace{2mm}

\textbf{Step 1:}  Compute $f^j_k$ for all $j \in [\gamma],\ k \in [d] \backslash \left\{0\right\}$, via \Cref{prop: theorem_WDD_background}. 

\vspace{1mm}

\textbf{Step 2:}  Reconstruct $f^j_0$ for all $j \in [\gamma]\backslash\{0\}$ via the linear system \eqref{eq: subspace 2}. 

\vspace{1mm}

\textbf{Step 3:}  Reconstruct $f^0_0$ via \eqref{eq: subspace 0 component}. 

\vspace{1mm}

\textbf{Step 4:} Build $X$ as in \eqref{eq: matrix of diagonals} by applying the inverse Fourier transform to the Fourier coefficients obtained in Step 1 - 3.

\vspace{1mm}
 
\textbf{Step 5:} Compute the magnitudes $\vert \tilde{x}_\ell\vert = \sqrt{X_{\ell,\ell}},~ \ell \in [d]$.

\vspace{1mm}

\textbf{Step 6:} Compute the top eigenvector $\tilde{z}$ of  $\mathrm{sgn}(X)$. Set $\mathrm{sgn}(\tilde{x}) = \mathrm{sgn}(\tilde{z})$.

\vspace{1mm}

\textbf{Output:} $\tilde{x} = \vert \tilde{x} \vert \cdot \mathrm{sgn}(\tilde{x}) \in \Complex^d$. 
\end{flushleft}

\end{algorithm}

\subsection{Reconstruction strategy and guarantees for phase objects} \label{sec: phase object algorithm}

In the following, we consider phase objects, which allow us to decouple the magnitude and the phase recovery for the lost coefficients $f_0^j,~ j \in [d]$. 
This class of objects is represented by the set
\[
\T^d \defeq \left\{ v \in \Complex^d: \vert v_j\vert = 1, ~ j \in [d]\right\}.
\]

Turning back to the recovery of $f^j_0$, the inclusion $x \in \mathbb T^d$ provides a straightforward procedure for magnitude recovery.

\begin{proposition} \label{prop: abs_zero_freq_phase_object}
Let $x \in \T^d$. For every $j \in [d]$, the diagonals $x \circ S_j \overline x$ belong to $\mathbb T^d$, and their Fourier transforms satisfy $\norm{f^j}_2^2 = d^2$. Consequently, $\vert f_0^j \vert$ can be recovered by 
\begin{equation*}
\vert f^j_0 \vert^2 = d^2 - \sum_{k=1}^{d-1} \vert f^j_k \vert^2.
\end{equation*}
\end{proposition}
\begin{proof}[Proof.]
For $x \in \T^d$, we have
\[
\left\vert \left(x \circ S_j\overline{x}\right)_{\ell}\right\vert = \vert x_{\ell}\vert \vert x_{\ell + j} \vert = 1,
\]
for all $ \ell, j \in [d]$. Plancherel's identity \eqref{thm: plancherel} provides
\begin{equation*} 
\sum_{k=0}^{d-1} \vert f^j_k \vert ^2 = \norm{f^j}_2^2 
= \left\|F\left[x \circ S_j\overline{x}\right]\right\|_2^2
= d \left\|x\circ S_j\overline{x}\right\|_2^2
= d \cdot \sum_{\ell = 0}^{d-1} \left\vert x_{\ell} x_{\ell + j} \right\vert^{2}
= d^2. 
\end{equation*}
\end{proof}

Next, the phase of $f^j_0$ has to be recovered. If, by \Cref{prop: abs_zero_freq_phase_object}, the resulting $\vert f_{0}^j \vert = 0$, the lost coefficient $f_{0}^j$ is zero. For all other cases, we need to find the argument of $f^j_0$, shortly denoted by 
\begin{equation}\label{eq: f phi definitions}
\varphi_j \defeq \arg(f^j_0).
\end{equation}

By \Cref{prop: abs_zero_freq_phase_object}, the diagonals satisfy $x \circ S_j\overline{x} \in \T^d$ for all $j \in [d]$. Rewriting this in terms of its Fourier coefficients gives
\begin{equation*} 
1 = \left\vert \left(x \circ S_j\overline{x}\right)_{\ell}\right\vert^2 =  \left\vert \left(F^{-1}  F \left[x \circ S_j\overline{x}\right]\right)_{\ell} \right\vert^2 
= \frac{1}{d^2} \bigg\vert \vert f^j_0 \vert \cdot e^{i \varphi_j} + \sum_{k=1}^{d-1} e^\frac{2\pi i \ell k}{d}  f^j_k \bigg\vert^2.
\end{equation*}
Hence, the missing argument $\varphi_j$ can be found by solving the following linear system.

\begin{theorem} \label{thm: linear_system}

Let $x \in \T^d$. 
For all $j \in [d]$ with $\vert f_{0}^j \vert > 0$, the argument $ \varphi_j \in (0,2\pi]$ as defined in \eqref{eq: f phi definitions} solves the system
\begin{align} \label{eq: linear_system}
    A^j \begin{pmatrix}
        \cos(\varphi_j)\\ \sin(\varphi_j)
    \end{pmatrix} 
= \frac{d^2}{2\vert f_0^j\vert}
    \begin{pmatrix} 1 - \frac{1}{d^2}\left(\vert a_0^j \vert^2 + \vert f_0^j\vert^2\right)\\
        \vdots\\
         1 - \frac{1}{d^2}\left(\vert a_{d-1}^j \vert^2 + \vert f_0^j\vert^2\right)
    \end{pmatrix},
\end{align}
where
\begin{align}
A^j \defeq \begin{pmatrix}
    \operatorname{Re}(a_0^j) & \operatorname{Im}(a_0^j)\\
        \vdots & \vdots \\
        \operatorname{Re}(a_{d-1}^j) & \operatorname{Im}(a_{d-1}^j)
    \end{pmatrix} \quad \text{and} \quad
a^j \defeq d F^{-1} 
\begin{pmatrix}
    0\\
    f^j_1 \\
    \vdots \\
    f^j_{d-1} \\
    \end{pmatrix}. \label{eq: def A a}
\end{align}
\end{theorem}
The proof of \Cref{thm: linear_system} can be found in \Cref{sec: proof thm linear system}. Essentially, it is not required to solve the, possibly large, linear system \eqref{eq: linear_system}. The $(k+1)$th equation of system \eqref{eq: linear_system} is equivalent to 
\begin{align*}
    \operatorname{Re}\left(a_k^j e^{-i \varphi_j}\right) = \frac{d^2 - \vert a_k^j\vert^2 - \vert f_0^j\vert^2}{2 \vert f_0^j \vert},
\end{align*}
where 
\begin{align*}
    \operatorname{Re}\left(a_k^j e^{-i \varphi_j}\right) = \vert a_k^j\vert \cos(\arg(a_k^j) - \varphi_j).
\end{align*}
Hence, we find that
\begin{align} \label{eq: recovery_formula}
    \varphi_j \in  \bigcap_{k=0}^{d-1} \left\{\arg(a^j_k) \pm \arccos\left( \frac{d^2-\vert a^j_k\vert^2-\vert f_0^j\vert^2}{2\vert a^j_k\vert\vert f_0^j\vert} \right)\right\}.
\end{align}
If $\rank(A^j) = 2$, the intersection of the sets in \eqref{eq: recovery_formula} has only one element. Computing these values for at least two distinct $k \in [d]$, we can determine the true value of the argument $\varphi_j$ of $f_0^j$.

Systems \eqref{eq: linear_system} provide only the recovery procedure for each $\varphi_j$ which is not necessarily unique. On the other hand, all $\varphi_j$ are linked to the same object $x$. Hence, by looking at multiple systems \eqref{eq: linear_system} at the same time, we are able to determine whether unique recovery of $x$ is possible or not. 
It turns out that only two families of objects cannot be recovered uniquely in the presence of background noise.

\begin{theorem}[Negative results] \label{thm: ambiguities}

$ $

\begin{enumerate}
    \item[(i)]  For all $m \in [d]$, the objects $x \sim x^m \defeq  \left(e^{\frac{2\pi i km}{d} }\right)_{k\in [d]}$ with appropriately chosen $b^m \in \R^d$ produce the same measurements \eqref{eq: ptycho_measurements}. 
    \item[(ii)] Suppose $d$ is even and consider $x^q$ for $q = 1,2$, defined by
    \begin{align*}
    x_k^q \defeq  e^{-\frac{2\pi i k m}{d}} \cdot \begin{cases}
        1, ~~k ~\text{even},\\ -(-1)^q e^{(-1)^q \frac{1}{2}i \rho}, ~~k ~\text{odd},
    \end{cases}
~ m \in [d], \ \rho \in (-\pi,\pi).
\end{align*}
If $x \sim x^1$, there exist backgrounds $b^q, \ q = 1,2,$ such that $(x^1,b^1)$ and $(x^2,b^2)$ result in the same measurements \eqref{eq: ptycho_measurements}.
\end{enumerate}
\end{theorem}

\Cref{thm: ambiguities} is proven in \Cref{sec: proof of ambiguity theorem}.

The first example corresponds to the case where the zeroth Fourier coefficients of all diagonals are the only nontrivial entries. Therefore, we have no means to reconstruct the phases of the lost coefficients. This happens only for modulations of a constantly one vector. 

The second case can be interpreted similarly to the conjugate reflection ambiguity well-known for Fourier phase retrieval \cite{beinert2015ambiguities}. In fact, we have
\begin{equation*}
\overline{x^{1}}_{-k} = (-1)^k x^2_k, \quad k\in [d].
\end{equation*}
However, this type of ambiguity only applies to objects described by $(ii)$ and not for all $x \in \T^d$.

For all other $x \in \T^d$ we can guarantee the following.
\begin{theorem}  \label{thm: unique_recovery}
Let \Cref{assumption: window} hold with $\gamma \geq 3$. Assume that $x \in  \T^d$ neither admits $(i)$ nor $(ii)$ in \Cref{thm: ambiguities}. Then, $x$ can be uniquely recovered from the measurements \eqref{eq: ptycho_measurements}.
\end{theorem}

The proof of \Cref{thm: unique_recovery} can be found in \Cref{sec: proof thm uniqueness}. 
Comparing \Cref{thm: unique_recovery} to \Cref{thm: bojarovska_flinth}, we observe that for measurements with background noise more diagonals are required to be taken into account to guarantee unique recovery. More precisely, the Fourier coefficients of the second off-diagonal $(j=2)$ have to be additionally considered to compensate for the coefficients lost due to the background noise. We also note that the zeroth diagonal only provides information about the magnitudes which are a priori known for phase objects. In contrary to the noise-free setting where unique reconstruction is guaranteed for every non-vanishing object, background noise causes two classes of phase objects that are non-vanishing and cannot be reconstructed uniquely. This holds true even if further diagonals are included.

As it was mentioned before, the proof relies upon the underlying rank of the linear systems \eqref{eq: linear_system}. We collect all possible scenarios in \Cref{table:1} and summarize a respective recovery strategy in \Cref{alg: algorithm 2}.
\begin{table}[h!]
\centering
\begin{tblr}{| c | c | c | c |c|} 
\hline 
 $\rank(A^1)$ & $\rank(A^2)$ & $d$ & Result & Corresponding Lemma \\ \hline
 2 &  &  & unique recovery  & \Cref{l: unique recovery from diagonal if coprime} with $j=1$\\  \hline
 0 & &  & negative example $(i)$ & \Cref{l: shift1_rank0} with $j = 1$ 
 \\ \hline 
 1 & 1 &  & not possible & \Cref{l: shift1_rank1_shift2_rank1} \\ \hline 
  1 & 0 & odd & not possible &  \Cref{l: negative result 2}
 \\ \hline 
 1 & 0 & even  & negative example $(ii)$ & \Cref{l: negative result 2} 
 \\ \hline 
  1 & 2 & odd & unique recovery & \Cref{l: unique recovery from diagonal if coprime} with $j=2$ \\ \hline 
    1 & 2 & even & unique recovery & \Cref{l: shift1_rank1_shift2_rank2} \\ \hline

\end{tblr}
\vspace{2mm}
\caption{The solvability of \eqref{eq: linear_system} in \Cref{thm: linear_system} depends on the rank of the matrix $A^j$. The listed cases have to be discussed in the proof of \Cref{thm: unique_recovery}. }
\label{table:1}
\end{table}

\RestyleAlgo{ruled}
\begin{algorithm}[t!] 
\caption{Recovery of a phase object from ptychographic measurements with background noise}\label{alg: algorithm 2}
\vspace{1mm}
\begin{flushleft}

\textbf{Input:} Ptychographic measurements $\tilde{Y} \in \R^{d\times d}$ as in  \eqref{eq: ptycho_measurements}.

\vspace{2mm}

\textbf{Step 1:} Compute $f^1_k$ for all $k \in [d] \backslash \left\{0\right\}$ via \Cref{prop: theorem_WDD_background}. 

Recover $\vert f_0^1\vert$ via \Cref{prop: abs_zero_freq_phase_object} and set up $A^1$.

\vspace{1mm}

    \uIf{$\rank(A^1) = 0$}{
    Stop. Ground-truth $x$ is of type $(i)$.
    }\uElseIf{$\rank(A^1) = 1$}{
    
    Compute $f^2_k$ for all $k \in [d]\backslash\left\{0\right\}$ via \Cref{prop: theorem_WDD_background}.
        
    Recover $\vert f_0^2\vert$ via \Cref{prop: abs_zero_freq_phase_object} and set up $A^2$.

        \uIf{$d$ is odd}{
        Solve \eqref{eq: linear_system} to recover $f^2$. 
        Go to Step 2 with $\gamma = 3$.
        }\uElse{
        Compute the two solutions $f^{1,q}, \ q = 1,2$, using \eqref{eq: recovery_formula}. 
        
            \uIf{$\rank(A^2) = 0$}{
            Stop. Ground-truth $x$ is of type $(ii)$.
            }\uElse{
            Solve \eqref{eq: linear_system} to recover $f^2$. Select $f^{1,q}, \ q = 1,2$, which admits equality in \eqref{eq: phase obj diag relation}. Go to Step 2 with $\gamma = 3$.
            }
        }
    }
    \uElse
    {Solve \eqref{eq: linear_system}. Go to Step 2 with $\gamma = 2$.}

 \vspace{1mm}

\textbf{Step 2:} Build $X$ as in \eqref{eq: matrix of diagonals} with $\gamma = 2$ or $\gamma = 3$  by applying the inverse Fourier transform to the Fourier coefficients obtained in Step 1.

\vspace{1mm}

\textbf{Step 3:} Compute the top eigenvector $\tilde{z}$ of  $\mathrm{sgn}(X)$. Set $\mathrm{sgn}(\tilde{x}) = \mathrm{sgn}(\tilde{z})$.

\vspace{2mm}

\textbf{Output:} $\tilde{x} = \mathrm{sgn}(\tilde{x}) \in \Complex^d$ with $\tilde{x} \sim x$.
\end{flushleft}
\end{algorithm}

\section{Numerical experiments} \label{sec: numerics}

We perform numerical experiments to test the performance of the suggested algorithms. While the theory we provide is restricted to a one-dimensional setting, we perform the experiments with 2D images. All algorithms can easily be extended to a two-dimensional version by replacing the corresponding operators with their 2D analogy \cite{cordor2020fast}. However, it remains an open question whether a uniqueness guarantee similar to \Cref{thm: unique_recovery} holds true. Our numerical experiments provide a first glimpse on this question.

In all trials, the algorithms are tested on synthetic data as depicted in \Cref{fig: general object and window}. The ground-truth object is a $128\times 128$ image with the argument and magnitude being the cameraman picture. For the phase object, we replace the magnitude with a matrix of ones. The localized window has a $16\times 16$ nonzero block, represented by a Gaussian bell with a random offset. The offset is required as \Cref{assumption: window} fails to hold for symmetric windows \cite{forstner2020well}.

\begin{figure}[b!]
    \centering
   \begin{minipage}[htbp]{0.32\textwidth}
\centering
    \includegraphics[width=1.5in]{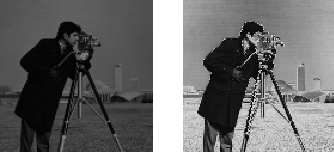}\\
    \figurecaptionfont (a) Magnitude and argument
\end{minipage}
\begin{minipage}[htbp]{0.32\textwidth}
\centering
    \includegraphics[width=0.685in]{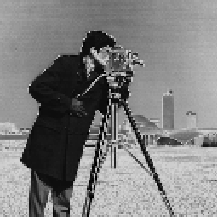}\\
    \figurecaptionfont (b) Argument
\end{minipage}
\begin{minipage}[htbp]{0.32\textwidth}
\centering
    \includegraphics[width=1.5in]{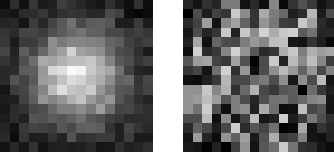}\\
    \figurecaptionfont (c) Magnitude and argument
 
\end{minipage}
    \caption{(a) Ground-truth object, (b) ground-truth phase object and (c) window.}
    \label{fig: general object and window}
\end{figure}

As background, we use the Shepp-Logan phantom in \Cref{fig: background and noisy diff pat} (a). \Cref{fig: background and noisy diff pat} (b) illustrates a diffraction pattern with background noise. Different levels of noise are simulated by adding different multiples of the phantom to the same diffraction patterns. To quantify the noise level, we use the ratio $\|Y - \tilde{Y}\|_F \big/ \norm{Y}_F,$
where $Y$ is as defined in \eqref{eq: ptychographic_measurements_noisefree} and $\tilde{Y}$ as in \eqref{eq: ptycho_measurements}.

\begin{figure}[t!]
    \centering
   \begin{minipage}[htbp]{0.3\textwidth}
\centering
    \includegraphics[width=1in]{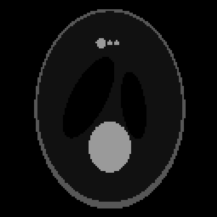}\\
    \figurecaptionfont (a) 
\end{minipage}
\begin{minipage}[htbp]{0.3\textwidth}
\centering
    \includegraphics[width=1in]{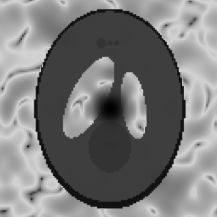}\\
    \figurecaptionfont (b) 
\end{minipage}
    \caption{(a) Background and (b) exemplary diffraction pattern perturbed by background noise in logarithmic scale.}
    \label{fig: background and noisy diff pat}
\end{figure}

As quality metrics, relative reconstruction error and relative measurement error 
\begin{equation*}
    \min_{\theta \in (0,2\pi]} \norm{x - e^{i\theta}\tilde{x}}_2 \big/ \norm{x}_2, \quad \text{and} \quad  \|Y - Y^{rec}\|_F \big/ \norm{Y}_F, 
\end{equation*}
are used, where $x$ is the ground-truth, $\tilde{x}$ is the reconstructed object and
$Y^{rec}_{\ell,r} = \left\vert \left( F \left[S_{-r}\tilde{x} \circ w\right] \right)_\ell\right\vert^2$, for all $\ell,r \in [d]$, are the simulated corresponding measurements.

In the following, we compare the object reconstruction obtained by the `vanilla' WDD algorithm (\Cref{alg: algorithm_WDD}) with the reconstruction results of our algorithms, described in \Cref{alg: algorithm general object} for general objects and in \Cref{alg: algorithm 2} for phase objects.

\begin{figure}[b!]
\centering
\includegraphics[width = 1\textwidth]
{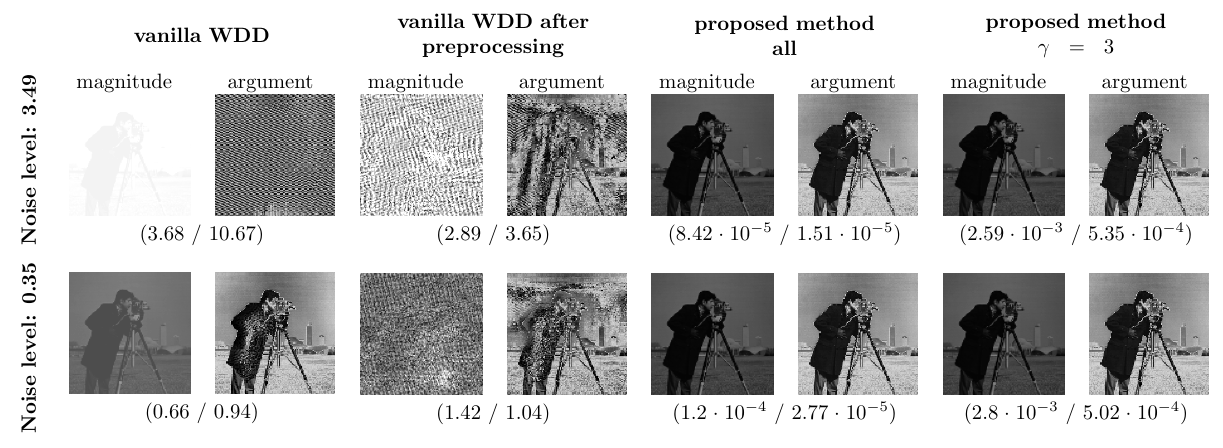}
\caption{Comparison of the performance of the vanilla WDD algorithm (\Cref{alg: algorithm_WDD}) without and with preprocessing, and \Cref{alg: algorithm general object} for different noise levels and different amounts of diagonals incorporated into the experiment (all or $\gamma = 3$). Below the reconstruction results, find the respective (relative reconstruction error/measurement error).}
\label{table: general object reco}
\end{figure}

In 2D, we interpret the parameter $\gamma$ in \Cref{assumption: window} for any $\gamma < \delta$ such that we consider all diagonals corresponding to tuples $j = (j_1,j_2)$ in 
\begin{equation*}
\q D_\gamma \defeq \{(j_1,j_2) \in \Z^2: 0 \leq j_1 < \gamma, \ -\gamma < j_2 < \gamma , \ -\gamma < j_1 + j_2 < \gamma\}.    
\end{equation*}
The case when all diagonals
\begin{equation*}
\{(j_1,j_2) \in \Z^2: 0 \leq j_1 < \delta, \ -\delta < j_2 < \delta\}
\end{equation*} 
are used is referred to as ``all''.

All experiments were performed in Python on a MacBook Pro equipped with a 2GHz Intel Core i5 and 16 GB of memory.

\begin{figure}[t!]
\centering
\includegraphics[width = 0.98 \textwidth]{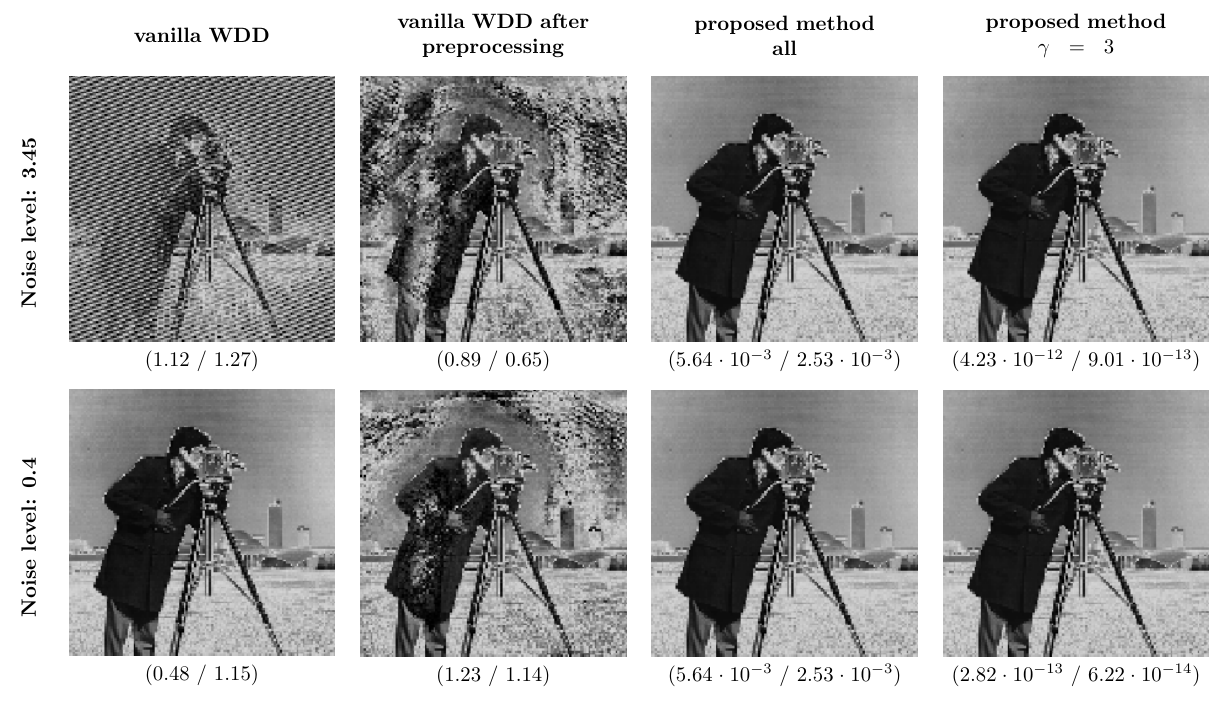}
\caption{Comparison of the performance of the vanilla WDD algorithm (\Cref{alg: algorithm_WDD}) without and with preprocessing, and \Cref{alg: algorithm 2} for different noise levels and different amounts of diagonals incorporated into the experiment (all or $\gamma = 3$). Below the reconstruction results, find the respective (relative reconstruction error / measurement error).}
\label{table: phase object reco}
\end{figure}

\begin{table}[b!]
\centering
\begin{tblr}{width=\textwidth,
        colspec={*{11}{X[c,m]}},
        vlines,
        column{1,2} = {2.5mm},
        column{3,6,9} = {7mm},
        column{4,7,10} = {11mm},
        column{5} = {7mm},
        cell{1}{1,2} = {r=2,c=1}{c},
        cell{3}{1} = {r=2,c=1}{c},
        cell{5}{1} = {r=2,c=1}{c},
        }
\hline 

 $d$ &  $\delta$ & \SetCell[c=3]{c}{{{run time  (in seconds)}}} & & & \SetCell[c=3]{c} {{{relative  reconstruction error}}} & & & \SetCell[c=3]{c} {{{relative measurement error}}} & & \\ \cline{3-11} 
   
  & & {{{WDD (all)}}} & {{{WDD ($\gamma = 3$)}}} & {{{ADP}}} &  {{{WDD (all)}}} & {{{WDD ($\gamma = 3$)}}} & {{{ADP}}} &  {{{WDD (all)}}} & {{{WDD ($\gamma = 3$)}}} & {{{ADP}}}  \\ \hline

\SetCell[r=2]{} \vspace{-3mm} 64
 & 8  & \centering 17 & \centering  7 & \centering 59   & \centering $1.6 \cdot 10^{-4}$ & \centering $3.2 \cdot 10^{-3}$  & \centering $0.29$ & \centering $5.1 \cdot 10^{-5}$ & \centering  $9.4 \cdot 10^{-4}$ &  \centering  0.09 \\\cline{2-11} 
 &  16  & \centering 62  & \centering 10  & \centering 64  & \centering $7.4 \cdot 10^{-5}$ & \centering  $2.1 \cdot 10^{-3}$ &  \centering  $0.30$ & \centering $1.6 \cdot 10^{-5}$ & \centering  $4.9 \cdot 10^{-4}$ &  \centering  0.15 \\ 
\hline
\SetCell[r=2]{} 96 
 &  8   & \centering 47 & \centering 41 & \centering 812  & \centering $2.5 \cdot 10^{-4}$  &  \centering  $1.7 \cdot 10^{-3}$   & \centering $0.38$& \centering $5.1 \cdot 10^{-5}$ & \centering  $5.5 \cdot 10^{-4}$ &  \centering  0.08 \\   \cline{2-11} 

 &  16  & \centering  143  & \centering 41  & \centering 789  & \centering  $5.5 \cdot 10^{-5}$ & \centering  $1.6 \cdot 10^{-3}$   & \centering $0.31$ & \centering $1.2 \cdot 10^{-5}$ & \centering  $4.0 \cdot 10^{-4}$ &  \centering  0.16 \\ 
 \hline

\end{tblr}
\vspace{1mm}
\caption{Comparison of \Cref{alg: algorithm general object} (WDD) with 10 iterations of the ADP algorithm \cite{chang2019advanced} for ptychographic measurements with background noise causing a noise level of size approximately $3.5$. }
\label{table: general object adp comparison}
\end{table}

First, we test \Cref{alg: algorithm general object} on the object in \Cref{fig: general object and window} (a). We compare the reconstructions obtained by our approach with the results of \Cref{alg: algorithm_WDD}. Apart from these two methods, we implemented the preprocessing approach by \cite{wang2017background}. 
The results can be found in \Cref{table: general object reco}.

As expected, vanilla WDD suffers from the noise. While preprocessing improves the reconstruction, it requires tuning the parameters for optimal performance depending on the noise level. In contrast, the method proposed by us filters the background noise and provides a good reconstruction independently of the noise level both when using all diagonals or only those in $\q D_3$.

Now, we apply \Cref{alg: algorithm 2} to reconstruct the phase object in \Cref{fig: general object and window} (b) from its ptychographic measurements with background noise. Again, the reconstruction results are compared with \Cref{alg: algorithm_WDD} with and without additional preprocessing \cite{wang2017background}. 
As these methods are not expected to produce a phase object, the reconstructions are projected on $\T^{d \times d}$.
\Cref{table: phase object reco} shows that, again, our proposed algorithm filters the background noise well. Interestingly, the $\mathcal D_3$ version of \Cref{alg: algorithm 2} shows better reconstruction errors while using less diagonals. This could be a result of a numerical error accumulation from a larger number of diagonals or a suboptimal outcome of the phase synchronization step.

Additionally, we compare the performance of our algorithms and the ADMM denoising algorithm for phase retrieval (ADP) proposed in \cite{chang2019advanced}. In \Cref{table: general object adp comparison} and \Cref{table: phase object adp comparison}, we list the runtime, the relative reconstruction and measurement errors of the respective methods. For general objects, an additional comparison of \Cref{alg: algorithm general object} both for all diagonals and for those in $\q D_3$ is included. We observe that WDD reconstructs the ground-truth object up to a numerical error. ADP, on the other hand, requires more iterations and time to match the same precision.

\begin{table}[t!]
\centering
\begin{tblr}{width=\textwidth,
        colspec={*{8}{X[c,m]}},
        vlines,
        column{1,2}={2.5mm},
        column{3,4,6,8}={10mm},
        cell{1}{1,2} = {r=2,c=1}{c},
        cell{3}{1} = {r=2,c=1}{c},
        cell{5}{1} = {r=2,c=1}{c},
        }
\hline 

 $d$ &  $\delta$ & \SetCell[c=2]{c}{{{run time (in seconds)}}} & & \SetCell[c=2]{c} {{{relative reconstruction error}}} & & \SetCell[c=2]{c} {{{relative measurement error}}} &  \\\cline{3-8} 
   
  & & {{{WDD}}} & {{{ADP }}} &  {{{WDD}}} &  {{{ADP }}} &  {{{WDD}}} &  {{{ADP }}}  \\ \hline

\SetCell[r=2]{} \vspace{-3mm} 64
 & 8    & \centering  5  & \centering 61  & \centering $2.7 \cdot 10^{-14}$   & \centering $0.95$ &  \centering  $1.0 \cdot 10^{-14}$ &  \centering  0.29 \\ \cline{2-8} 
 &  16    & \centering 8  & \centering 58 & \centering  $5.7 \cdot 10^{-14}$ &  \centering  $0.80$  &  \centering  $2.8 \cdot 10^{-14}$ &  \centering  0.45   \\ 
\hline
\SetCell[r=2]{} 96 
 &  8    & \centering 33 & \centering 705  &  \centering   $1.9 \cdot 10^{-13}$  & \centering $0.99$ &  \centering  $3.4 \cdot 10^{-14}$ &  \centering  0.27  \\   \cline{2-8} 

 &  16     & \centering  37 & \centering 778  & \centering  $1.1 \cdot 10^{-13}$   & \centering 0.93  &  \centering  $3.8 \cdot 10^{-14}$ &  \centering  0.44 \\ 
 \hline

\end{tblr}
\vspace{1mm}
\caption{Comparison of \Cref{alg: algorithm 2} (WDD with $\gamma = 3$) with 10 iterations of the ADP algorithm \cite{chang2019advanced} for ptychographic measurements of a phase object with background noise causing a noise level of size approximately $3.5$. }
\label{table: phase object adp comparison}
\end{table}

\section{Proofs} \label{sec: proofs}

\subsection{Proof of  \Cref{thm: linear_system}} \label{sec: proof thm linear system}

For $x \in \T^d$, by \Cref{prop: abs_zero_freq_phase_object} we have
\begin{equation} \label{eq: relation_inv_F_diag_1}
\left\vert \left(F^{-1} \left( F \left[x \circ S_j\overline{x}\right]\right)\right)_{\ell} \right\vert 
= \left\vert \left(x \circ S_j\overline{x}\right)_{\ell}\right\vert 
= 1,
\end{equation}
for all $ \ell, j \in [d]$. We expand the inverse Fourier transform as
\begin{equation}  \label{eq: relation_inv_F_diag_2}
    \left(F^{-1} \left( F \left[x \circ S_j\overline{x}\right]\right)\right)_{\ell} = \frac{1}{d} \sum_{k=0}^{d-1} e^\frac{2\pi i \ell k}{d}  \left( F \left[x \circ S_j\overline{x}\right]\right)_{k}
    =\frac{1}{d} 
     \left\vert f_0^j \right\vert e^{i \varphi_j} + \frac{1}{d}
     a_\ell^j,
\end{equation}
where the definitions \eqref{eq: f phi definitions} and \eqref{eq: def A a} were used.

To obtain $\varphi_j$, we combine \eqref{eq: relation_inv_F_diag_1} and \eqref{eq: relation_inv_F_diag_2}, and study the quadratic equation
\begin{equation*}
    d^2 = \left\vert \vert f_{0}^j \vert \cdot e^{i \varphi_j} +a_{\ell}^j \right\vert^2, ~~ \ell \in [d].
\end{equation*}
Expanding the squares gives 
\begin{align*}
    \left\vert \vert f_{0}^j \vert \cdot e^{i \varphi_j} +a_{\ell}^j \right\vert^2 &=  \vert f_{0}^j \vert^2 + 2\vert f_{0}^j \vert\operatorname{Re}\left(a_\ell^je^{-i\varphi_j}\right) + \vert a_{\ell}^j \vert^2\\[6pt]
    &=  \vert f_{0}^j \vert^2 + 2\vert f_{0}^j \vert \left(\operatorname{Re}(a_\ell^j) \cos(\varphi_j) + \operatorname{Im}(a_\ell^j)\sin(\varphi_j)\right) + \vert a_{\ell}^j \vert^2,
\end{align*}
and we rewrite the condition to be satisfied by $\varphi_j$ as
\begin{align*}
  \operatorname{Re}(a_\ell^j) \cos(\varphi_j) + \operatorname{Im}(a_\ell^j)\sin(\varphi_j) = \frac{1}{2\vert f_{0}^j \vert} \left(   d^2 -  \vert a_{\ell}^j \vert^2 - \vert f_{0}^j \vert^2\right)
\end{align*}
for all $\ell \in [d]$. That means, we need to solve (\ref{eq: linear_system}) to recover $\varphi_j$.
\qed

\subsection{Proof of \Cref{thm: ambiguities}} \label{sec: proof of ambiguity theorem}

Firstly, we show that, for all $m \in [d]$, the objects $x \sim \left(e^{\frac{2\pi i \ell m}{d} }\right)_{\ell \in [d]}$ with appropriately chosen background $b^m \in \R^d$ produce the same measurements. 

Let $x^m \defeq   \left(e^{\frac{2\pi i \ell m}{d}}\right)_{\ell\in [d]}$ for $m\in [d]$. Let $\tilde{Y}^{m} \in \R^{d\times d}$ be the respective measurements \eqref{eq: ptycho_measurements} for all $x^{m}$, $m\in [d]$,  with some background noise $b^{m} \in \R^d$. 

For all objects $x^m$, all diagonals are constant, more precisely
\begin{equation} \label{eq: modulation objects constant diagonals}
    (x^m\circ S_j\overline{x^m})_\ell = e^{-\frac{2\pi i j m}{d}}
\end{equation}
for all $\ell \in [d]$. Hence, for all $m,j \in [d]$, the respective Fourier coefficients $f^{j,m}_k$ are zero for all $k \in [d]\backslash \left\{0\right\}$. Consequently, for all $m,j \in [d],$ and all $k\in [d]\backslash \left\{0\right\}$, \Cref{prop: theorem_WDD_background} provides $(F^{-1}\tilde{Y}^{0}F)_{j,k} = (F^{-1}\tilde{Y}^{m}F)_{j,k}$. 
For $k = 0$, \Cref{prop: theorem_WDD_background} tells that $(F^{-1}\tilde{Y}^{0}F)_{j,0} = (F^{-1}\tilde{Y}^{m}F)_{j,0}$ for all $j\in [d]$ if and only if 
\begin{equation*} 
\left(f^{j,0}_0 - f^{j,m}_0\right) (\overline{F[\overline{w} \circ S_j w]})_0 = d(F^{-1}(b^{m} - b^{0}))_j,
\end{equation*}
which is equivalent to
\begin{equation*}
    (F^{-1}(b^{m} - b^{0}))_j = \left(1 - e^{-\frac{2\pi i j m}{d}} \right)  (\overline{F[\overline{w} \circ S_j w]})_0,
\end{equation*}
as from \eqref{eq: modulation objects constant diagonals} we obtain  $f^{j,m}_0 = de^{-\frac{2\pi i j m}{d}}$.

Via this relation, we can define backgrounds $b^m$ which cause the same measurements for objects $x^m$. For example, we can set $(F^{-1} b^0)_j = 0$ for all $j \in [d]\backslash \left\{0\right\}$, and choose $(F^{-1} b^0)_0 \in \R$ suitably. 

As $F^{-1} b^0$ is chosen conjugate symmetric, its Fourier transform $b^0$ is real-valued according to \Cref{l: fourier of real vector conjugate symmetric}. Moreover, $F^{-1} b^m$ is conjugate symmetric since  $(F^{-1} b^m)_0 = (F^{-1} b^0)_0 \in \R$ as $e^{-\frac{2\pi i j m}{d}} = 1$ for $j = 0$, and
\begin{align*}
    \overline{(F^{-1} b^m)_{-j}}& = \overline{\left(1 - e^{\frac{2\pi i j m}{d}} \right)}  (F[\overline{w} \circ S_{-j} w])_0\\
    &= \left(1 - e^{-\frac{2\pi i j m}{d}} \right)  \overline{(F[\overline{w} \circ S_{j} w])}_0 = (F^{-1}b^m)_j,
\end{align*}
for all $j \in [d]\backslash \left\{0\right\}$, where we use the symmetry of 
\begin{equation} \label{eq: symmetry of wigner window}
    (F[\overline{w} \circ S_{-j} w])_0 = \sum_{\ell = 0}^{d-1} \overline{w}_\ell w_{\ell - j} = \sum_{\ell = 0}^{d-1} \overline{w}_{\ell + j} w_\ell =  \overline{(F[\overline{w} \circ S_{j} w])}_0.
\end{equation}
Hence, also $b^m \in \R^d$ for all $m \in [d]$.

It remains to choose $(F^{-1} b^0)_0 \in \R$ large enough such that the backgrounds $b^m$ satisfy $b^m_\ell \geq - \vert F[S_{-r} x^m \circ w]_\ell \vert^2 $, to ensure $\tilde{Y}^m_{\ell,r} \geq 0$ for all $m, \ell, r \in [d]$. The inverse Fourier transform gives
\[
(F^{-1} b^0)_0 
= (F^{-1} b^m)_0  
\geq - \vert F[S_{-r} x^m \circ w]_\ell \vert^2 - \sum_{j=1}^{d-1} e^{\frac{2 \pi i j \ell}{d}} (F^{-1} b^m)_j .
\]
Taking the maximum on the right-hand side over $\ell, m \in [d]$ provides a suitable choice of $(F^{-1} b^0)_0$.

We conclude that, with such background noise, objects $x^m$ produce the same measurements \eqref{eq: ptycho_measurements} and, thus, cannot be recovered uniquely.

Secondly, we show that for objects $x^q, \ q = 1,2,$ defined by entries
\begin{align} \label{eq: object_type_2}
    x_\ell^q = e^{-\frac{2\pi i \ell m}{d}} \cdot \begin{cases}
        1, \quad &\ell ~\text{even},\\-(-1)^q e^{(-1)^q \frac{1}{2}i \rho}, &\ell ~\text{odd},
    \end{cases}, \quad \ell \in [d],
\end{align}
there exist backgrounds $b^q, \ q = 1,2,$ such that $(x^1,b^1)$ and $(x^2,b^2)$ result in the same measurements \eqref{eq: ptycho_measurements}, i.e., there exist $b^1, b^2 \in \R^{ d}$ with 
\begin{equation*}
   \tilde{Y}^1_{\ell,r} \defeq \left\vert \left( F \left[S_{-r}x^1 \circ w\right]\right)_\ell \right\vert^2 + b^1_\ell =  \left\vert \left( F \left[S_{-r}x^2 \circ w\right]\right)_\ell \right\vert^2 + b^2_\ell \defeqr \tilde{Y}^2_{\ell,r},
\end{equation*}
for all $\ell,r \in [d]$.

Fix $m \in \Z, \ \rho \in (-\pi,\pi)$, and consider objects $x^q, \ q = 1,2$, defined as in \eqref{eq: object_type_2}. The diagonals corresponding to the objects $x^q, \ q = 1,2,$ are given by
\begin{align*}
   ( x^q \circ S_j \overline{x^q})_\ell  = \begin{cases}
       1, \quad &j ~ \text{even},\\
-(-1)^q e^{\frac{2\pi i j m}{d} -(-1)^q (-1)^\ell \frac{1}{2}i\rho },  \quad &j ~ \text{odd},
   \end{cases} \quad \ell \in [d].
\end{align*}
For $j$ even, the Fourier coefficients are
\begin{align} \label{eq: coefficients even}
    f^{j,q}_k  = \begin{cases}
       d, \quad & k = 0,\\
0,  \quad & k \in [d] \backslash \{0\},
   \end{cases}
\end{align}
and for $j$ odd, they are 
\begin{align*} 
    f^{j,q}_k &= -(-1)^q e^{\frac{2\pi i j m}{d}} \left( \sum_{\ell = 0}^{\frac{d}{2}-1} e^{-\frac{2\pi i k \frac{\ell}{2}}{\frac{d}{2}}} e^{-(-1)^q \frac{1}{2}i\rho} +  \sum_{\ell = 0}^{\frac{d}{2}-1} e^{-\frac{2\pi i k \frac{\ell-1}{2}}{\frac{d}{2}}} e^{(-1)^q \frac{1}{2}i\rho}  \right) \notag \\
    &= \begin{cases}
       - (-1)^q e^{\frac{2\pi i j m}{d}} \cdot \frac{d}{2} \cos(\frac{\rho}{2}) , \quad & k = 0,\\
  e^{\frac{2\pi i j m}{d}} \cdot \frac{d}{2} i \sin(\frac{\rho}{2}) , \quad & k = \frac{d}{2},\\     
0,  \quad & k \in [d] \backslash \{0,\frac{d}{2}\}.
   \end{cases}
\end{align*}
Hence, for all $j \in [d]$ and all  $k \in [d] \backslash \{0\}$, we know $f^{j,1}_k = f^{j,2}_k$, so that \Cref{prop: theorem_WDD_background} guarantees $(F^{-1}\tilde{Y}^{1}F)_{j,k} = (F^{-1}\tilde{Y}^{2}F)_{j,k}$. For $k = 0$, again, it holds $(F^{-1}\tilde{Y}^{1}F)_{j,0} = (F^{-1}\tilde{Y}^{2}F)_{j,0}$ for all $j\in [d]$ if and only if 
\begin{equation*} 
\left(f^{j,1}_0 - f^{j,2}_0\right) (\overline{F[\overline{w} \circ S_j w]})_0 = d(F^{-1}(b^{2} - b^{1}))_j.
\end{equation*}
Analogously to the above, we can conclude that there exists background noise with which objects $x^1$ and $x^2$ produce the same measurements \eqref{eq: ptycho_measurements}, i.e., unique recovery of objects $x^1$ or $x^2$ is also not possible.
\qed

\subsection{Analysis of linear systems \eqref{eq: linear_system}} \label{sec: disjoint_analysis}
In this section, the foundation of the proof of \Cref{thm: unique_recovery} is established by studying the linear systems \eqref{eq: linear_system}. 
For any $j \in [d]$, the three possible cases $\rank(A^j) = 0, \ \rank(A^j) = 1,$ and $\rank(A^j) = 2$ are investigated.

In the following, we write $F^{-1} f^j$ instead of $x \circ S_j \overline{x}$ as there is not necessarily a unique $x$ corresponding to $F^{-1} f^j$.

We start with the case $\rank(A^j) = 0$.

\begin{lemma} \label{l: shift1_rank0}
Let $j \in [d]$. The following claims are equivalent:
\begin{enumerate}
    \item[(i)] The matrix $A^j$ has $\rank(A^j) = 0$.
    \item[(ii)] $(F^{-1} f^j)_k = e^{i \varphi_j}$ for all $k \in [d]$, where $\varphi_j$ is as in \eqref{eq: f phi definitions}. 
    \item[(iii)] The zeroth Fourier coefficient admits $\vert f^j_0 \vert = d$.
\end{enumerate} 
Furthermore, if $\rank(A^j) = 0$ and $j$ is coprime with $d$,  the ground-truth object is $x \sim \left(e^{\frac{2\pi i km}{d}}\right)_{k\in [d]}$, for some $m \in [d]$.
\end{lemma}

\begin{proof}[Proof.]
$(i) \Rightarrow (ii), (iii):$ Let $\rank(A^j) = 0$. Definition \eqref{eq: def A a} gives $a_\ell^j = 0$ for all $\ell \in [d]$. We then obtain 
\begin{equation*}
d (F^{-1} f^j)_\ell = \sum_{k=0}^{d-1} e^{\frac{2\pi i k\cdot \ell}{d}}f^j_k = a_\ell^j  + f_0^j = f_0^j = \vert f_0^j \vert e^{i \varphi_j},
\end{equation*}
for all $\ell \in [d]$, i.e., $F^{-1} f^j$ has constant value $f_0^j /d$. By \Cref{prop: abs_zero_freq_phase_object}, $F^{-1} f^j \in \T^d$ and, thus, $F^{-1} f^j = e^{i \varphi_j}$ and $\vert f_0^j \vert = d$.  

$(ii) \Rightarrow (iii):$ If $(F^{-1} f^j)_\ell = e^{i \varphi_j} \in \T$ for all $\ell \in [d]$, we get 
\begin{align*}
   f^j_k = (F(F^{-1} f^j))_k = \sum_{\ell=0}^{d-1} e^{\frac{-2\pi i k \cdot \ell}{d}} e^{i \varphi_j} =  \begin{cases}
d e^{i \varphi_j}, \quad &k = 0, \\ 0, \quad &k \in [d]\backslash\{0\},
\end{cases} 
\end{align*}
and, hence, $\vert f^j_0\vert = d \vert e^{i \varphi_j} \vert = d$.

$(iii) \Rightarrow (i):$ If $\vert f^j_0\vert = d$, by \Cref{prop: abs_zero_freq_phase_object}, the definition in \eqref{eq: def A a}, and Plancherel's identity \eqref{thm: plancherel} we have
\[
\norm{a^j}_2^2 
= \tfrac{1}{d} \norm{F a^j}_2^2 
= d \norm{(0,f_1^j,\ldots, f_{d-1}^j)}_2^2 = d \norm{f^j}_2^2 - d \vert f^j_0 \vert^2
= d^3 - d^3 = 0.
\]
That means, $a^j_k =0$ for all $k \in [d]$, and $\rank(A^j) = 0$.

In addition to $\rank(A^j) = 0$, let us assume that $j$ is coprime with $d$. Consider any $x\in \T^d$ with $x \circ S_j \overline{x} = F^{-1} f^j$. From the equivalence of $(i)$ and $(ii)$ proven above, we find that  $x_k \overline{x}_{k+j} = (F^{-1} f^j)_k = e^{i\varphi_j}$ for all $k\in [d]$, i.e.,
      \begin{equation} \label{eq: relation phases of x}
          \arg(x_k) - \arg(x_{k+j}) = \varphi_j + 2\pi m
      \end{equation}
      for some $m \in \Z$.
With this, we obtain
\begin{equation*}
          \arg(x_{\ell j})  \in  
          \arg(x_{0}) - \ell \varphi_j + 2\pi \Z
      \end{equation*}
      for any $\ell \in \Z$. Since $j$ and $d$ are coprime, $\ell = d$ is the smallest integer satisfying $\ell j = 0 \MOD d$. We conclude that  
\begin{equation*}
          \arg(x_0)   = \arg(x_{dj})  \in  \arg(x_{0}) - d \varphi_j + 2\pi \Z,
      \end{equation*}
hence, $\varphi_j \in \frac{2\pi}{d}  \Z$. Combining this with \eqref{eq: relation phases of x}, we find that the corresponding objects are of the type $x \sim  \left(e^{\frac{2\pi i km}{d}}\right)_{k\in [d]}$, for some $m \in \Z$.
\end{proof}

Next, we study the system \eqref{eq: linear_system} in case $\rank(A^j) = 1$.

\begin{lemma} \label{l: shift_1_rank_1_no_unique_recovery}
Let $j \in [d]$ and $\rank(A^j) = 1$. There exist two solutions $\varphi^j_1 \neq \varphi^j_2$ to \eqref{eq: linear_system} given by
\begin{align*}
    \varphi^j_q \defeq \arg(a_0^j) + (-1)^{q} \arccos\left(\frac{d^2-\vert a_0^j\vert^2-\vert f_0^j\vert^2}{2\vert a_0^j\vert \vert f_0^j\vert}\right), \quad q =1,2.
\end{align*}
Furthermore, there exists a disjoint partition $\mathcal S_j, \mathcal S_j^c$ of $[d]$ such that $0 \in \mathcal S_j$, $\mathcal S_j^c \neq \varnothing$, and we can write the corresponding diagonals $F^{-1} f^{j,q}$ as
\begin{equation*}
(F^{-1} f^{j,q})_k = \begin{cases}
e^{i\psi^j_q}, & k \in \mathcal{S}_j, \\
e^{i(\psi^j_q + (-1)^{q} \rho^j)}, & k \in \mathcal{S}_j^c, 
\end{cases}
\quad q =1,2,
\end{equation*}
with $\rho^j \in (-\pi,\pi) \backslash \{0\}$, $\psi_1^j \neq \psi_2^j,$ and 
\begin{equation} \label{eq: relation_psi1_psi2}
    \psi_2^j \in  \psi_1^j - \rho^j + \pi + 2\pi \Z.
\end{equation}

\end{lemma}

\begin{proof}[Proof.]

\Cref{l: shift_1_rank_1_no_unique_recovery} can be interpreted geometrically as a problem of intersecting circles. For a visualization of the proof idea, see \Cref{fig: visualization_intersecting_circles}.

\begin{figure}[!ht] 
\centering
\begin{tikzpicture}

\draw[lightgray, thick, loosely dotted] (1.96,0.4) -- (2.66,1.39);
\draw[lightgray, thick, loosely dotted] (-0.28,1.98) -- (0.42,2.97);

\draw[lightgray, thick, loosely dotted] (1.96,0.4) -- (-0.42,-2.97);
\draw[lightgray, thick, loosely dotted] (-0.28,1.98) --  (-2.66,-1.39);

\filldraw [black] (0,0) circle (1pt) node[anchor=north]{\tiny$0$};

\draw[color=TUMBlau!80, thick](0,0) circle (2);
\filldraw [black] (0,-2) circle (0.1pt)
node[anchor=north]{\tiny$\left\{ z\in \Complex: \vert z \vert = \vert f_0 \vert \right\}$};

\draw[color=TUMGreen!80,  thick](0,0) circle (3);
\filldraw[black] (3,0) circle (0.1pt)
node[anchor=west]{\tiny$\left\{ z\in \Complex: \vert z \vert = d\right\}$};

\draw[color=TUMOrange!80, thick](-0.7,-0.99) circle (3);
\filldraw [black] (-0.32,-3.97) circle (0.1pt)
node[anchor=north]{\tiny$\left\{ z \in \Complex: \vert z - a_0 \vert = d \right\}$};

\filldraw [black] (-0.28,1.98) circle (1pt) node[anchor=south]{\tiny$f_{0}^{1}$};
\filldraw [black] (1.96,0.4) circle (1pt) node[anchor=west]{\tiny$f_{0}^{2}$};

\filldraw [black] (0.42,2.97) circle (1pt) node[anchor=south]{\tiny $f_0^1 + a_0 = d e^{i\psi_1}$};
\filldraw [black] (2.66,1.39) circle (1pt)node[anchor=south]{\tiny $f_0^2 + a_0 = d e^{i\psi_2}$};

\filldraw [black] (-2.66,-1.39) circle (1pt) node[anchor=north east]{\tiny $f_0^1 + c a_0 = d e^{i(\psi_1 - \rho)}$};
\filldraw [black] (-0.42,-2.97) circle (1pt) node[anchor=north]{\tiny $f_0^2 + c a_0 = d e^{i(\psi_2 + \rho)}$};

\draw[gray, thick,  densely dotted, ->] (1.25,1.71)  arc (54:96:2.1);
\draw[gray, thick,  densely dotted, ->] (1.26,1.7)  arc (54:11:2.1);

\filldraw [black] (1.15,1.63) circle (1pt) node[anchor = east]{\tiny $\vert f_{0}\vert e^{i \arg(a_0)}$};

\draw[gray, thick , densely dotted, ->] (0.45,3.07) arc (82:207:3.12);

\draw[gray, thick , densely dotted, ->] (2.75,1.43)  arc (27:-98:3.12);

\end{tikzpicture}
\caption{Geometric visualization of \Cref{l: shift_1_rank_1_no_unique_recovery}. By \Cref{prop: abs_zero_freq_phase_object}, $\vert f_0^q\vert$ is known. This guarantees that the possible Fourier coefficients $f^q_0$ lie on the blue circle. By definition, $f_0^{q} + a_\ell = d (F^{-1} f^q)_\ell$ and $\vert (F^{-1} f^q)_\ell \vert = 1 $. Hence, $f_0^q + a_\ell$ has to lie on the green circle. Consequently, the possible solutions $f^q_0$ lie on the orange circle, or, more precisely, on its intersections with the blue circle. It can be shown that neither the blue and the green circle coincide nor the orange circle is tangent to the blue circle, so that there are precisely two intersection points, corresponding to two solutions. The lightgray dotted lines are parallels in the direction of $a_0$ containing the solutions $f^q_0$. Their two intersection with the green circle show that there are two distinct values the entries of the diagonal can take. 
}
\label{fig: visualization_intersecting_circles}
\end{figure}

In the following, let $j \in [d]$ be fixed. Throughout the proof, we drop the index $j$ to simplify the notation.

Since $\rank(A) = 1$, all equations in \eqref{eq: linear_system} are the same up to a multiplicative constant $c_k \in \R,$ i.e., 
\begin{equation} \label{eq: ak = c a0}
    a_k = c_k \cdot a_0,
\end{equation} 
for all $k \in [d]\backslash\{0\}.$ If, for some $k \in [d]\backslash\{0\}$, we have $\vert a_k \vert = 0$, the corresponding equation 
\[
\RE(a_k) \cos(\varphi) + \IM(a_k)\sin(\varphi) 
= \frac{1}{2\vert f_{0} \vert} \left( d^2 -  \vert a_k \vert^2 - \vert f_{0} \vert^2\right)
\]
in \eqref{eq: linear_system} reduces to $d = \vert f_0 \vert$. By \Cref{l: shift1_rank0}, this is equivalent to $\rank(A)=0$, which contradicts $\rank(A)=1$. Hence, $\vert a_k \vert > 0$ and $c_k \neq 0$ for all $k \in [d]\backslash\{0\}$.

We can rewrite the equations of \eqref{eq: linear_system} as
\begin{equation*}
    d^2 = \vert a_{k} \vert^2 + \vert f_{0} \vert^2 + 2\vert f_{0} \vert\operatorname{Re}\left(a_k e^{-i\varphi}\right).
\end{equation*}
Since \eqref{eq: linear_system} is only considered for $\vert f_0\vert > 0$, it is further transformed into
\begin{equation*}
    \cos(\arg(a_k) - \varphi) = \frac{d^2-\vert a_k\vert^2 - \vert f_0\vert^2}{2\vert a_k\vert\vert f_0\vert},
\end{equation*}
which may have two solutions
\begin{equation*}
\varphi_q = \arg(a_k) + (-1)^{q} \arccos\left(\frac{d^2-\vert a_k\vert^2-\vert f_0\vert^2}{2\vert a_k\vert \vert f_0\vert}\right), \quad q =1,2.
\end{equation*}
As the system \eqref{eq: linear_system} has at least one solution corresponding to the ground-truth object $x$, the right-hand side is the same for all $k \in [d]$. Setting $k=0$ gives the desired formula for $\varphi_q$, $q = 1,2$.

Note that $\varphi_1 = \varphi_2$ if and only if, for all $k \in [d]$,
\begin{equation*}
     \frac{d^2-\vert a_k\vert^2 - \vert f_0\vert^2}{2\vert a_k\vert\vert f_0\vert} = \delta,
\end{equation*}
for $\delta \in  \left\{-1,1\right\}$. Bringing the denominator to the other side and summing the equations for all $k$ yields
\begin{equation}\label{eq: tech rank 1}
d^3 - \sum_{k=0}^{d-1}\vert a_k\vert^2 - d\vert f_0\vert^2 
= 2\delta  \vert f_0\vert \sum_{k=0}^{d-1}\vert a_k\vert.
\end{equation} 
Plancherel's identity \eqref{thm: plancherel} combined with \eqref{eq: def A a} and \Cref{prop: abs_zero_freq_phase_object} gives
\begin{equation*}
\norm{a}_2^2 
= \tfrac{1}{d} \norm{F a}_2^2 
= d \norm{(0,f_1,\ldots, f_{d-1})}_2^2 = d \norm{f}_2^2 - d \vert f_0 \vert^2
= d^3 - d \vert f_0 \vert^2.
\end{equation*}
Substituting this into \eqref{eq: tech rank 1} results in
\begin{equation*}
 2 \vert f_0\vert \sum_{k=0}^{d-1}\vert a_k\vert = 0. 
\end{equation*}
Recall that the case $\vert f_0\vert = 0$ is excluded as it does not require solving the linear system \eqref{eq: linear_system}. Moreover, we showed above that all $a_k \neq 0$. Hence, in the case $\rank(A) = 1$, there always exist two values $\varphi_1 \neq \varphi_2$ solving \eqref{eq: linear_system}. In the following, the corresponding two vectors of Fourier coefficients are denoted by $f^{1}$ and $f^{2}$, respectively.

These solutions can be described in more detail. Let us set $\psi_q$ so that
\begin{equation*}
\left(F^{-1}f^{q}\right)_0 = e^{i\psi_q}.
\end{equation*}
By the definitions \eqref{eq: def f}, \eqref{eq: def A a}, and \eqref{eq: ak = c a0}, we obtain
\begin{equation*}
f_0^{q} + c_k a_0 
= f_0^{q} + a_k 
= \sum_{\ell = 0}^{d-1} e^{\frac{2\pi i k\ell}{d}} f^{q}_\ell
=d (F^{-1} f^{q})_k.
\end{equation*}
By \Cref{prop: abs_zero_freq_phase_object}, the diagonals satisfy $F^{-1} f^{q} \in \T^d$.  Thus, we have that $\tfrac{1}{d}(f_0^{q} + c_k a_0 ) \in \T$.

The system \eqref{eq: linear_system} has at least one solution corresponding to the ground-truth object $x$. Consequently, the right-hand side of \eqref{eq: linear_system} belongs to 
\[
\operatorname{im}(A) \defeq \{v \in \R^{d}: A u = v \ \text{for some} \ u \in \R^2 \}. 
\]
For every $v \in \operatorname{im}(A)$ we have
\[
v_k = \RE(a_k) u_1 + \IM(a_k) u_2
= c_k \RE(a_0) u_1 + c_k \IM(a_0) u_2 = c_k v_0,
\]
and, thus, the right-hand side of \eqref{eq: linear_system} also admits 
\begin{equation} \label{eq: rhs_k = c rhs_0}
1 - \frac{1}{d^2}\left(\vert a_{k} \vert^2 + \vert f_0\vert^2\right) 
= c_k \left[1 - \frac{1}{d^2}\left(\vert a_0 \vert^2 + \vert f_0\vert^2\right) \right].    
\end{equation}

Inserting \eqref{eq: ak = c a0} into \eqref{eq: rhs_k = c rhs_0}, we obtain
\begin{equation*}
    d^2 - c_k^2 \vert a_0\vert^2 - \vert f_0\vert^2 = c_k \cdot \left( d^2 - \vert a_0 \vert ^2 - \vert f_0\vert ^2\right).
\end{equation*}
This quadratic equation with respect to $c_k$ has two roots, $1$ and $c \defeq (\vert f_0\vert^2-d^2)/\vert a_0\vert^2 < 0$, noticeably, independent of $k$. Let us define $\mathcal{S}$ and $\mathcal{S}^c$ as
\[
\mathcal{S} \defeq \{ k \in [d] : c_k = 1\}
\quad \text{and} \quad
\mathcal{S}^c \defeq \{ k \in [d] : c_k = c \}.
\]
By construction, $0 \in \mathcal{S}$, $\mathcal{S} \cap \mathcal{S}^c = \varnothing,$ and $\mathcal{S} \cup \mathcal{S}^c= [d]$. If $\mathcal{S}^c$ is empty, the entries of $(F^{-1} f^{q})$ are all equal to $\tfrac{1}{d}(f_0^{q} +  a_0)$, which, by \Cref{l: shift1_rank0}, contradicts the $\rank(A)=1$ assumption. Thus, we conclude that $\mathcal{S}^c \neq \varnothing$.       

Then, for both, $q=1$ and $q=2$, the values of $F^{-1} f^{q}$ are given by 
\[
(F^{-1} f^{q})_k 
= \tfrac{1}{d}(f_0^{q} + a_0) 
= e^{i\psi_q}, \quad k \in \mathcal{S},  
\]
and
\[
(F^{-1} f^{q})_k 
= \tfrac{1}{d}(f_0^{q} + c a_0) 
\defeqr e^{i(\psi_q +\rho_q)}, \quad k \in \mathcal{S}^c,
\]
for some $\rho_q \in (-\pi, \pi] \backslash \{0\}$.

The next step is to show that $\rho_2 = -\rho_1$. We have
\begin{equation*}
e^{i(\psi_q+\rho_q)} 
= \tfrac{1}{d}(f_0^{q} + c a_0 + a_0 - a_0)
= e^{i\psi_q} + (c - 1) \cdot \frac{1}{d}a_0, \quad q = 1,2.
\end{equation*}
Combining the two relations for $q = 1$ and $q = 2$ yields 
 \begin{equation} \label{eq: difference_of_values}
   e^{i\psi_1} \left( 1 - e^{i\rho_1}\right) =  e^{i\psi_2} \left( 1 - e^{i\rho_2}\right).
\end{equation}
Since $\vert e^{i\psi_1} \vert = \vert e^{i\psi_2} \vert = 1$, we get 
\begin{equation*}
   \vert 1 - e^{i\rho_1}\vert^2 = \vert 1 - e^{i\rho_2}\vert^2,
\end{equation*}
or, equivalently, 
\begin{equation*}
   \cos \rho_1 = \cos \rho_2.
\end{equation*}
As we choose $\rho_1, \rho_2 \in (-\pi,\pi] \backslash\left\{0\right\}$, either $\rho_2 = \rho_1$ or $\rho_2 = -\rho_1$ holds. If $\rho_2 = \rho_1$, then \eqref{eq: difference_of_values} yields $\psi_1 = \psi_2$ and $f^{1}_0 = f^{2}_0$, which in turn, implies $\varphi_1 = \varphi_2$ and contradicts $\varphi_1 \neq \varphi_2$. We conclude that $\rho_2 = -\rho_1$.

In particular, for $\rho_1 = \pi$, the only $\rho_2 \in (-\pi,\pi]\backslash\{0\}$ satisfying $\cos(\rho_2) = \cos(\rho_1) = -1$ is $\rho_2 = \pi = \rho_1$, contradicting $\rho_2 \neq \rho_1$. Therefore, we have $\rho_1 \in (-\pi,\pi)\backslash \{0\}$. 
    
Finally, via \eqref{eq: difference_of_values}, we obtain   
\begin{equation*}
    e^{i\psi_1} = e^{i\psi_2} \cdot \frac{e^{-i\rho_1} - 1}{e^{i\rho_1} - 1}
    = e^{i\psi_2} \cdot (-e^{-i\rho_1})
    = e^{i\psi_2} \cdot e^{-i\rho_1} \cdot e^{i\pi},
\end{equation*}
so that
\begin{equation*}
    \psi_1 = \psi_2 - \rho_1 + \pi + 2\pi m
\end{equation*}
for some $m \in \Z$. Setting $\rho = \rho_2 = -\rho_1$, we obtain the statement of the lemma.
\end{proof}

The remaining case to consider is $\rank(A^j) = 2$. In this case, system \eqref{eq: linear_system} is uniquely solvable. Hence, the lost zero frequency $f_0^j =  \left(F\left[x \circ S_j \overline{x}\right]\right)_0$ can be recovered uniquely. Then, $x$ can be recovered uniquely if, additionally, $j$ is coprime with the object's dimension $d$.

\begin{lemma} \label{l: unique recovery from diagonal if coprime}
Let $j \in [d]$. 

\begin{itemize}
    \item[(i)] If $\rank(A^j) = 2$, $\varphi_j$ is uniquely recoverable from \eqref{eq: linear_system}, i.e., $F^{-1} f^j $ is uniquely recoverable from measurements \eqref{eq: ptycho_measurements}.

    \item[(ii)] If $j$ is coprime with $d$, $x$ can be uniquely recovered from the diagonal $F^{-1} f^j = x \circ S_j \overline{x}$.
\end{itemize}
\end{lemma}

\begin{proof}[Proof.]
If $\rank(A^j) = 2$, system (\ref{eq: linear_system}) has one solution, i.e., we can recover $\varphi_j$, and, hence, $ \left(F\left[x \circ  S_j\overline{ x}\right]\right)_0$ uniquely. From  $F\left[x \circ  S_j\overline{ x}\right]$, we obtain back $x \circ  S_j\overline{ x}$.

If $j$ and $d$ are coprime, $j$ is a generator of the additive group of integers modulo $d$, denoted by $\Z_d$. That means, for every $k \in [d]$, there exists $r \in \Z$ with $rj \MOD d = k$. Hence, for any $k \in [d]$, there exists $r \in \Z$ with
\begin{equation*}
    x_0 \overline{x}_k = \prod_{t=1}^{r} x_{(t-1)j}\overline{x}_{tj} \ \bigg/ \ \prod_{t=1}^{r-1} \vert x_{tj}\vert^2.
\end{equation*}
Consequently, for all $k \in [d]$, the product $ x_0 \overline{x}_k$ can be built from the entries of $x \circ S_j \overline{x}$ together with the magnitudes of $x$, which are a priori known for a phase object. 
The phase of $x_0 \in \T$ can be chosen arbitrarily as the solution to a phase retrieval problem is unique only up to a global phase. Based on this choice, the phases of $x_k \in \T$ for all $k>0$ are found using the relations $ x_0 \overline{x}_k$. 
\end{proof}

\begin{remark}\label{rem: rank 2 coprime}
If $j$ is not coprime with $d$, $j$ does not generate $\Z_d$. That means, $\mathcal{C}_0 \defeq \left\{j\ell \MOD d, \ell \in [d]\right\} \neq \Z_d$, and there exists at least one $n \in \Z_d$ such that  $\mathcal C \defeq \{(j \ell + n) \MOD d, \ell \in [d]\}$ satisfies $\mathcal{C}_0 \cap \mathcal{C} = \varnothing$. Hence, for every object $x \in \bb C^d$ we can construct $\tilde x$ as 
\[
\tilde x_k = 
\begin{cases}
\alpha x_k, & k \in \mathcal C, \\
x_k, & k \notin \mathcal C,
\end{cases}
\]
with $\alpha \in \T~\backslash \left\{1\right\}$, such that $\tilde{x} \not\sim x$ but $x\circ S_j\overline{x} = \tilde{x} \circ S_j \overline{\tilde{x}}$.

\end{remark}

\subsection{Proof of \Cref{thm: unique_recovery}} \label{sec: proof thm uniqueness}

In this section, we prove \Cref{thm: unique_recovery}. Using the results of \Cref{sec: disjoint_analysis}, we investigate in closer detail under which conditions system (\ref{eq: linear_system}) is uniquely solvable. We will find that it is sufficient to use the information given by $j = 1$ and $j = 2$ to fully characterize the uniqueness of reconstruction for measurements with background noise. The proof is split into the analysis of the disjoint cases listed in \Cref{table:1}. 

The case $j = 0$ is of no relevance as 
\begin{equation*}
     \left(F\left[x \circ  S_0\overline{ x}\right]\right)_0 = \sum_{k = 0}^{d-1} e^{-\frac{2\pi i k\cdot 0}{d}} \cdot \vert x_{k}\vert^2 = d
\end{equation*}
is known for a phase object, and does not provide further information on the phases of the object.

We start with investigating $j = 1$. As $1$ is coprime with any $d \in \N$, \Cref{l: unique recovery from diagonal if coprime} guarantees unique recovery in case $\rank(A^1) = 2$.

Let $\rank(A^1) = 0$. With $1$ being coprime with any $d \in \N$, \Cref{l: shift1_rank0} provides that the only object satisfying this condition is the negative example $(i)$ in \Cref{thm: ambiguities}. In \Cref{sec: proof of ambiguity theorem} we showed that for this class of objects unique recovery is not possible.

If  $\rank(A^1) = 1$, \Cref{l: shift_1_rank_1_no_unique_recovery} tells that considering only the diagonal corresponding to shift $j=1$ is not sufficient. Thus, we need to include information for further diagonals.

We continue with $j = 2$ and use that the first and the second off-diagonal are related via 
\begin{equation}\label{eq: phase obj diag relation}
\begin{aligned}
(F^{-1} f^{2})_k &= (x \circ S_2 \overline x)_k 
= x_k \cdot 1 \cdot \overline{x_{k+2}} 
= x_k  ~\vert x_{k+1} \vert ^2~ \overline{x_{k+2}} \\
& = (x \circ S_1 \overline x)_k ~ (x \circ S_1 \overline x)_{k+1} 
= (F^{-1} f^{1})_k ~(F^{-1} f^{1})_{k+1}.
\end{aligned}
\end{equation}
Again, the three possible cases $\rank(A^2) = 0$, $ \rank(A^2) = 1$, and $\rank(A^2) = 2$ are considered.

Firstly, we find that $ \rank(A^2) = 1$ is not feasible if $\rank(A^1) = 1$. 

\begin{lemma} \label{l: shift1_rank1_shift2_rank1}
The case $\rank(A^1) = 1$ and $\rank(A^2) = 1$ is not possible.
\end{lemma}

\begin{proof}[Proof.]
We show the statement by contradiction. Assume that $\rank(A^1) = \rank(A^2) = 1$. As $\rank(A^1) = 1$, \Cref{l: shift_1_rank_1_no_unique_recovery} for $j=1$ states that there are the two distinct solutions 
\begin{align*}
     (F^{-1} f^{1,q})_k = \begin{cases}
     e^{i\psi_q}, & k \in \mathcal S, \\
     e^{i(\psi_q + (-1)^{q}\rho)}, & k \in \mathcal S^c,    
     \end{cases}
     \quad q = 1,2,
\end{align*}
with $\mathcal S, \mathcal S^c$ as described in the proof of \Cref{l: shift_1_rank_1_no_unique_recovery}.

Here, we drop the index $j$ for $\psi^j_q$ and $\rho^j$ to shorten the notation as we only require these values for $j = 1$. However, we keep the index in $f^{j,q}$ to distinguish $f^{1,q}$ and $f^{2,q}$.

By \eqref{eq: phase obj diag relation}, there are three possible values which the entries of the second diagonal can take, 
\begin{align}  \label{eq: second diagonal elements}
    (F^{-1} f^{2,q})_k 
    &=(F^{-1} f^{1,q})_k (F^{-1} f^{1,q})_{k+1} \notag\\
    &= \begin{cases} e^{i2\psi_q},  &k \in \mathcal{S},\ k+1 \in \mathcal{S}, \\
     e^{i(2\psi_q + (-1)^{q} \rho)}, &k \in \mathcal{S},\ k+1 \in \mathcal{S}^c \\ &\text{or}~  k \in \mathcal{S}^c,\ k+1 \in \mathcal{S},\\
     e^{i\left(2\psi_q + (-1)^{q} 2 \rho\right)},  &k \in \mathcal{S}^c, \ k+1 \in \mathcal{S}^c, \end{cases}
\end{align}
for both $ q= 1$ and $q = 2$. However, since $\rank(A^2) = 1$, by \Cref{l: shift_1_rank_1_no_unique_recovery}, the second diagonal can only accept two different values. The index sets $\mathcal S, \mathcal S^c \neq \varnothing$, hence there is at least one $k \in [d]$ with 
\begin{equation*}
\left(F^{-1}f^{1,q}\right)_k = e^{i\psi_q} \quad \text{and} \quad \left(F^{-1}f^{1,q}\right)_{k+1} = e^{i(\psi_q + (-1)^{q}\rho)}.
\end{equation*}
Thus, we can conclude for the second diagonal that
\begin{equation*}
(F^{-1}f^{2,q})_k = e^{i\psi_q} \cdot  e^{i(\psi_q + (-1)^{q} \rho)} = e^{i(2\psi_q + (-1)^{q} \rho)}    
\end{equation*}
for at least one $k \in [d]$. Consequently,
\begin{equation*}
        (F^{-1} f^{2,q})_k \in \left\{
     e^{i(2\psi_q + (-1)^{q} \ell \rho)}, e^{i\left(2\psi_q + (-1)^{q} (\ell+1)\rho\right)}\right\}, \quad q= 1,2,
\end{equation*}
where in both cases $\ell$ is either zero or one, depending on the set $\mathcal S$.

From \Cref{l: shift_1_rank_1_no_unique_recovery} for $j=1$, we obtain 
\begin{equation} \label{eq: relation_psi1_psi2_A1}
    \psi_1 = \psi_2 + \rho + \pi + 2\pi m_1,
\end{equation}
and, for $j=2$,
 \begin{equation} \label{eq: relation_psi1_psi2_A2}
    2\psi_1 - \ell \rho = 2\psi_2 + \ell \rho + \rho + \pi + 2\pi m_2
\end{equation}
for some $m_1, m_2 \in \Z$.

Bringing $\psi_1 - \psi_2$ to the left--hand side and the rest to the right-hand side in both \eqref{eq: relation_psi1_psi2_A1} and \eqref{eq: relation_psi1_psi2_A2} yields
\begin{equation*} 
 \rho + \pi + 2\pi m_1 = \frac{1}{2}\left( (2 \ell +1) \rho + \pi + 2\pi m_2 \right),        \end{equation*}
so that
\begin{equation*} 
(2 \ell - 1) \rho = \pi + 2\pi(2m_1 - m_2).        
\end{equation*}
For both, $\ell = 0$ and $\ell = 1$, we find that $\rho \in \left\{ -\pi,\pi\right\}$, which is not attainable according to \Cref{l: shift_1_rank_1_no_unique_recovery}. Hence, $\rank(A^1) = \rank(A^2) = 1$ is not feasible.
\end{proof}

Next, we investigate the case $\rank(A^1) = 1$ and $\rank(A^2) = 0$ and find the second negative example $(ii)$ in \Cref{thm: ambiguities}. From \Cref{sec: proof of ambiguity theorem} it is known that such objects cannot be recovered uniquely.

\begin{lemma} \label{l: negative result 2}
Let $x^q \in \Complex^d$ be defined by entries
\begin{align*}
    x_k^{q} \defeq e^{-\frac{2\pi i k m}{d}} \cdot \begin{cases}
        1, \quad &k ~\text{even},\\-(-1)^q e^{(-1)^q \frac{1}{2}i \rho}, &k ~\text{odd},
    \end{cases} 
\end{align*}
for all $k \in [d]$. If $\rank(A^1) = 1$ and $\rank(A^2) = 0$, then $d$ is even and $x \sim x^q$ with $q = 1$ or $q = 2$ for some $m \in \Z$ and $ \rho \in (-\pi,\pi)\backslash\{0\}$. 
\end{lemma}

\begin{proof}[Proof.]    

According to \Cref{l: shift1_rank0}, $\rank(A^2) = 0$ is equivalent to $F^{-1} f^2$ constant.

We assume $\rank(A^1) = 1$. As in \Cref{l: shift_1_rank_1_no_unique_recovery}, set $(F^{-1} f^{1,q})_0 \defeq e^{i\psi_q}$ for some $\psi_q \in [0,2\pi), \ q = 1,2$, and the other value that appears for at least one $(F^{-1} f^{1,q})_k$ we set as $e^{i(\psi_q +(-1)^q \rho)}$ for $\rho \in (-\pi,\pi)\backslash\{0\}$, for both $q=1,2$. 

The first and the second off-diagonal are related via \eqref{eq: phase obj diag relation}. Hence, based on the knowledge of the first diagonals $F^{-1}f^{1,q}$ and the fact that $F^{-1} f^2$ is constant, the second diagonal must be $(F^{-1} f^2)_k = e^{i(2\psi_q +(-1)^q \rho)}$ for all $k \in [d]$, for either $q=1$ or $q=2$. For this to hold true, the entries of $F^{-1} f^{1,q}$ have to satisfy
\begin{align} \label{eq: first diagonal}
    \left(F^{-1} f^{1,q}\right)_k = \begin{cases}
        e^{i\psi_q}, \quad &k ~\text{even},\\ e^{i(\psi_q +(-1)^q \rho)}, &k ~\text{odd}.
    \end{cases}
\end{align}

If $d$ is odd, this means $x_{d-1}\overline{x}_0 = e^{i\psi_q}$, so 
\begin{equation*}
x_{d-1}\overline{x}_1 = x_{d-1}\overline{x}_0x_0\overline{x}_1 = e^{i2\psi_q}.    
\end{equation*}
This requires $\rho \in 2\pi \Z$. Then, the first diagonal is $ \left(F^{-1} f^{1,q}\right)_k = e^{i\psi_q}$ for all $k \in [d]$, contradicting $\rank(A^1) = 1$. Hence, $\rank(A^1) = 1$ and $\rank(A^2) = 0$ can only appear if the object dimension $d$ is even.

Let $d$ be even.
By 
\begin{equation*}
    \prod_{k=0}^{d-1} x_k \overline{x}_{k+1} = \prod_{k=0}^{d-1} \vert x_k\vert^2 = 1, 
\end{equation*}
we obtain
\begin{equation*}
    \left(e^{i\psi_1}\right)^{\frac{d}{2}}\left(e^{i(\psi_1-\rho)}\right)^{\frac{d}{2}} = 1,
\end{equation*}
i.e., 
\begin{equation} \label{eq: relation_psi1_rho}
    d\psi_1 - \frac{d}{2}\rho = 2\pi m
\end{equation}
for some $m \in \Z$. Equation \eqref{eq: relation_psi1_psi2} in \Cref{l: shift_1_rank_1_no_unique_recovery} further provides that the respective $\psi_2$ corresponding to $f^{1,2}$ satisfies $\psi_2 = \psi_1 - \rho - \pi + 2\pi \tilde{m}$ for some $\tilde{m} \in \Z$. 

The first diagonal is sufficient to reconstruct $x$ as explained in \Cref{l: unique recovery from diagonal if coprime}. From \eqref{eq: first diagonal}, we derive the two possible solutions
\begin{align*}
    x_k^{q} =  e^{-ik(\psi_{q} + (-1)^{q} \rho)} \cdot \begin{cases}
        e^{(-1)^q i \frac{k}{2} \rho}, \quad & k ~\text{even},\\ e^{(-1)^q i \frac{k+1}{2} \rho}, & k ~\text{odd},
    \end{cases} \quad q = 1,2.
\end{align*}
Together with \eqref{eq: relation_psi1_psi2} and \eqref{eq: relation_psi1_rho},
we obtain that $x$ can be equal to $x^q$ for $q = 1$ or $q = 2$ with 
\begin{align*}
    x_k^{q} =  e^{-\frac{2\pi i k m}{d}} \cdot \begin{cases}
        1, \quad &k ~\text{even},\\-(-1)^q e^{(-1)^q \frac{1}{2}i \rho}, &k ~\text{odd},
    \end{cases} \quad q = 1,2,
\end{align*}
for all $k \in [d]$, for some $m \in \Z$ and $ \rho \in (-\pi,\pi)\backslash\{0\}$. 
\end{proof}

\begin{remark}
    In \Cref{thm: ambiguities} $(ii)$, we do not exclude $\rho = 0$. Note that for $\rho = 0$, the objects of type $(ii)$ equal the objects of type $(i)$. In \Cref{l: negative result 2}, however, $\rho = 0$ is excluded as it is not attainable in case $\rank(A^1) = 1$.
\end{remark}

The remaining case to investigate is $\rank(A^2) = 2$. 
If  $\rank(A^2) = 2$ and $d$ is odd, $d$ is coprime with $2$ and the ground-truth solution $x$ can be
uniquely recovered as shown in \Cref{l: unique recovery from diagonal if coprime} for $j =2$. It is left to study the case when $d$ is even.

\begin{lemma} \label{l: shift1_rank1_shift2_rank2}
If $\rank(A^1) = 1$, $\rank(A^2) = 2,$ and $d$ is even, the ground-truth solution $x$ can be uniquely recovered.

\end{lemma}

\begin{proof}[Proof.]  

If  $\rank(A^2) = 2$, the system (\ref{eq: linear_system}) for $j=2$ has one solution. That means, we can recover $f^2_0$ uniquely and compute the second diagonal $F^{-1} f^2$.

Since $d$ is even, $d$ and $j = 2$ are not coprime and there are multiple $x$ corresponding to the second diagonal $F^{-1} f^2$, see \Cref{rem: rank 2 coprime}. However, we can make use of $f^2$ to select the ground-truth solution out of the two possible solutions $f^{1,1} \neq f^{1,2}$ obtained in \Cref{l: shift_1_rank_1_no_unique_recovery} for the shift $j=1$. We show that this is possible by contradiction.

Assume that $f^{1,1}$ and $f^{1,2}$ yield, via \eqref{eq: phase obj diag relation}, the same diagonal $F^{-1} f^2$. The entries of $F^{-1} f^2$ can be deduced from $f^{1,1}$ and $f^{1,2}$ and can take the three possible values as in \eqref{eq: second diagonal elements} for both $q = 1$ and $q = 2$. 

It is not possible that $F^{-1}f^{2}$ is constant.  Otherwise, $\rank(A^2) = 0$ according to \Cref{l: shift1_rank0}. Thus, for at least one $k \in [d]$, 
\begin{equation*}
(F^{-1}f^{2})_k \in \left\{ e^{i2\psi_q}, e^{i(2\psi_q +(-1)^q 2\rho)} \right\}    
\end{equation*}
for both $ q= 1,2$, and, as we assumed that $F^{-1} f^2$ is obtained from both $f^{1,1}$ and $f^{1,2}$, either
\[
(i) \quad e^{i2\psi_1} =  e^{i2\psi_2} 
\quad \text{or} \quad 
(ii) \quad e^{i(2\psi_1 - 2\rho)}  =  e^{i(2\psi_2 +2 \rho)}
\]
must be satisfied. Suppose $(i)$ holds true. Then 
\begin{equation*}
    2\psi_1 = 2\psi_2 + 2\pi m
\end{equation*}
for some $m \in \Z$, i.e.,
\begin{equation*}
    \psi_1 = \psi_2 + \pi m.
\end{equation*}
For $m$ even this means $\psi_1 = \psi_2$ since $\psi_1,\psi_2 \in (0,2\pi]$, which is impossible by \Cref{l: shift_1_rank_1_no_unique_recovery}.
If $m$ is odd, we have
\begin{equation*}
\psi_1 = \psi_2 + \pi + 2\pi \tilde{m}    
\end{equation*}
with $\tilde{m} \in \Z$. Together with \eqref{eq: relation_psi1_psi2}, we obtain $\rho \in 2\pi \Z$, which is again not possible by \Cref{l: shift_1_rank_1_no_unique_recovery}.

Next, suppose property $(ii)$ is true and
\begin{align*}
    2\psi_1 - 2\rho = 2\psi_2 + 2\rho + 2\pi m_1
\end{align*}
for some $m_1 \in \Z$.  Combining this with \eqref{eq: relation_psi1_psi2}, we further find that $\rho \in  \pi\Z.$
Once again, this is not attainable according to \Cref{l: shift_1_rank_1_no_unique_recovery}. 

We conclude that both $(i)$ and $(ii)$ are not possible and we got a contradiction. Hence, only one out of $f^{1,1}$ or $f^{1,2}$ is compatible with the uniquely recovered $f^2$.
\end{proof}

In summary, we found that involving $j = 2$ helps to distinguish the two possible diagonals caused by $\rank(A^1) = 1$ into the ground-truth diagonal and the false diagonal. Knowing the first diagonal determines $x$ uniquely, see \Cref{l: unique recovery from diagonal if coprime} $(ii)$. 

We conclude that only from $j = 1$ and $j = 2$ it can be decided whether $x$ can be recovered uniquely, up to the global phase, and which objects can never be uniquely recovered from data with background noise. \Cref{thm: unique_recovery} summarizes these results.

\section{Conclusion and discussion}\label{sec: conclusions}

In this paper, we considered ptychographic measurements corrupted with background noise. Along the lines of the Wigner Distribution Deconvolution approach for ptychography, we designed two denoising algorithms, one for arbitrary objects and another version for phase objects. For the latter algorithm, a uniqueness guarantee was established for almost every object.

Following up on this approach, it would be interesting to investigate whether discarding more frequencies can be offset by the redundancy and how this affects the uniqueness of reconstruction.

Another promising direction is to use the established analysis for subspace completion technique \cite{forstner2020well}. It requires solving a system alike to \eqref{eq: subspace_completion_linear_system} with only one unknown coefficient, which is more difficult than the phase objects, but easier than two unknowns in the case of arbitrary objects. This can be seen as an intermediate step for establishing the uniqueness of reconstruction for general objects in the presence of background noise.

\bmhead{Acknowledgments}  The authors would like to thank Tim Salditt and Christian Schroer for providing valuable insights into the physics of imaging experiments.
Furthermore, the authors are grateful for the helpful comments on this work by Benedikt Diederichs and Frank Filbir.

\section*{Declarations}

\bmhead{Funding}
The authors acknowledge support by the Helmholtz Association under contracts No.~ZT-I-0025 (Ptychography 4.0), No.~ZT-I-PF-4-018 (AsoftXm), No.~ZT-I-PF-5-28 (EDARTI), No.~ZT-I-PF-4-024 (BRLEMMM).

\bibliography{sn-bibliography.bbl} 

\begin{thebibliography}{10}
\providecommand{\url}[1]{{#1}}
\providecommand{\urlprefix}{URL }
\providecommand{\doi}[1]{\url{https://doi.org/#1}}
\bibcommenthead

\bibitem{hoppe1969beugung}
W.~Hoppe, {Beugung im inhomogenen Prim{\"a}rstrahlwellenfeld. I. Prinzip einer
  Phasenmessung von Elektronenbeungungsinterferenzen}.
\newblock Acta Crystallographica Section A: Crystal Physics, Diffraction,
  Theoretical and General Crystallography \textbf{25}(4), 495--501 (1969)

\bibitem{pfeiffer2018x}
F.~Pfeiffer, X-ray ptychography.
\newblock Nature Photonics \textbf{12}(1), 9--17 (2018)

\bibitem{rodenburg2019ptychography}
J.~Rodenburg, A.~Maiden, Ptychography.
\newblock Springer Handbook of Microscopy pp. 819--904 (2019)

\bibitem{giewekemeyer2010quantitative}
K.~Giewekemeyer, P.~Thibault, S.~Kalbfleisch, A.~Beerlink, C.M. Kewish,
  M.~Dierolf, F.~Pfeiffer, T.~Salditt, Quantitative biological imaging by
  ptychographic x-ray diffraction microscopy.
\newblock Proceedings of the National Academy of Sciences \textbf{107}(2),
  529--534 (2010)

\bibitem{shi2019x}
X.~Shi, N.~Burdet, B.~Chen, G.~Xiong, R.~Streubel, R.~Harder, I.K. Robinson,
  X-ray ptychography on low-dimensional hard-condensed matter materials.
\newblock Applied Physics Reviews \textbf{6}(1) (2019)

\bibitem{hawkes2019springer}
P.W. Hawkes, J.C. Spence, \emph{{Springer handbook of microscopy}} ({Springer
  Nature}, Cham, 2019)

\bibitem{candes2015phase}
E.J. Candes, X.~Li, M.~Soltanolkotabi, {Phase retrieval via Wirtinger flow:
  Theory and algorithms}.
\newblock IEEE Transactions on Information Theory \textbf{61}(4), 1985--2007
  (2015)

\bibitem{xu2018accelerated}
R.~Xu, M.~Soltanolkotabi, J.P. Haldar, W.~Unglaub, J.~Zusman, A.F. Levi, R.M.
  Leahy, {Accelerated Wirtinger flow: A fast algorithm for ptychography}.
\newblock arXiv preprint arXiv:1806.05546  (2018)

\bibitem{wang2017solving}
G.~Wang, G.B. Giannakis, Y.C. Eldar, Solving systems of random quadratic
  equations via truncated amplitude flow.
\newblock IEEE Transactions on Information Theory \textbf{64}(2), 773--794
  (2017)

\bibitem{melnyk2022stochastic}
O.~Melnyk, {Stochastic Amplitude Flow for phase retrieval, its convergence and
  doppelgängers}.
\newblock arXiv preprint arXiv:2212.04916  (2022)

\bibitem{gerchberg1972practical}
R.~Gerchberg, W.~Saxton, A practical algorithm for the determination of phase
  from image and diffraction plane picture.
\newblock Optik \textbf{35}, 237--246 (1972)

\bibitem{fienup1978reconstruction}
J.R. Fienup, Reconstruction of an object from the modulus of its fourier
  transform.
\newblock Optics letters \textbf{3}(1), 27--29 (1978)

\bibitem{marchesini2016alternating}
S.~Marchesini, Y.C. Tu, H.T. Wu, Alternating projection, ptychographic imaging
  and phase synchronization.
\newblock Applied and Computational Harmonic Analysis \textbf{41}(3), 815--851
  (2016)

\bibitem{luke2004relaxed}
D.R. Luke, Relaxed averaged alternating reflections for diffraction imaging.
\newblock Inverse problems \textbf{21}(1), 37 (2004)

\bibitem{chang2018total}
H.~Chang, Y.~Lou, Y.~Duan, S.~Marchesini, {Total variation-based phase
  retrieval for Poisson noise removal}.
\newblock SIAM Journal on Imaging Sciences \textbf{11}(1), 24--55 (2018)

\bibitem{chang2019blind}
H.~Chang, P.~Enfedaque, S.~Marchesini, {Blind ptychographic phase retrieval via
  convergent alternating direction method of multipliers}.
\newblock SIAM Journal on Imaging Sciences \textbf{12}(1), 153--185 (2019)

\bibitem{rodenburg2004phase}
J.M. Rodenburg, H.M. Faulkner, A phase retrieval algorithm for shifting
  illumination.
\newblock Applied physics letters \textbf{85}(20), 4795--4797 (2004)

\bibitem{melnyk2023convergence}
O.~Melnyk, Convergence properties of gradient methods for blind ptychography.
\newblock arXiv preprint arXiv:2306.08750  (2023)

\bibitem{rodenburg1992theory}
J.~Rodenburg, R.~Bates, {The theory of super-resolution electron microscopy via
  Wigner-distribution deconvolution}.
\newblock Philosophical Transactions of the Royal Society of London. Series A:
  Physical and Engineering Sciences \textbf{339}(1655), 521--553 (1992)

\bibitem{chapman1996phase}
H.N. Chapman, {Phase-retrieval X-ray microscopy by Wigner-distribution
  deconvolution}.
\newblock Ultramicroscopy \textbf{66}(3-4), 153--172 (1996)

\bibitem{iwen2016fast}
M.A. Iwen, A.~Viswanathan, Y.~Wang, Fast phase retrieval from local correlation
  measurements.
\newblock SIAM Journal on Imaging Sciences \textbf{9}(4), 1655--1688 (2016)

\bibitem{iwen2020phase}
M.A. Iwen, B.~Preskitt, R.~Saab, A.~Viswanathan, {Phase retrieval from local
  measurements: Improved robustness via eigenvector-based angular
  synchronization}.
\newblock Applied and Computational Harmonic Analysis \textbf{48}(1), 415--444
  (2020)

\bibitem{preskitt2018phase}
B.P. Preskitt, Phase retrieval from locally supported measurements.
\newblock Ph.D. thesis, University of California, San Diego (2018)

\bibitem{cordor2020fast}
C.~Cordor, B.~Williams, Y.~Hristova, A.~Viswanathan, in \emph{{28th European
  Signal Processing Conference (EUSIPCO 2020)}}, ed. by A.~Marques, B.~Hunyadi
  (IEEE, [Piscataway, NJ], 2020), pp. 980--984

\bibitem{perlmutter2020provably}
M.~Perlmutter, N.~Sissouno, A.~Viswantathan, M.~Iwen, in \emph{{28th European
  Signal Processing Conference (EUSIPCO 2020)}}, ed. by A.~Marques, B.~Hunyadi
  (IEEE, [Piscataway, NJ], 2020), pp. 970--974

\bibitem{perlmutter2021inverting}
M.~Perlmutter, S.~Merhi, A.~Viswanathan, M.~Iwen, {Inverting spectrogram
  measurements via aliased Wigner distribution deconvolution and angular
  synchronization}.
\newblock Information and Inference: A Journal of the IMA \textbf{10}(4),
  1491--1531 (2021)

\bibitem{melnyk2023phase}
O.~Melnyk, {Phase Retrieval from Short-Time Fourier Measurements and
  Applications to Ptychography}.
\newblock Ph.D. thesis, Technische Universit{\"a}t M{\"u}nchen (2023)

\bibitem{bojarovska2016phase}
I.~Bojarovska, A.~Flinth, {Phase retrieval from Gabor measurements}.
\newblock Journal of Fourier Analysis and Applications \textbf{22}(3), 542--567
  (2016)

\bibitem{jaganathan2016stft}
K.~Jaganathan, Y.C. Eldar, B.~Hassibi, {STFT phase retrieval: Uniqueness
  guarantees and recovery algorithms}.
\newblock IEEE Journal of selected topics in signal processing \textbf{10}(4),
  770--781 (2016)

\bibitem{alaifari2021stability}
R.~Alaifari, M.~Wellershoff, {Stability estimates for phase retrieval from
  discrete Gabor measurements}.
\newblock Journal of Fourier Analysis and Applications \textbf{27}, 1--31
  (2021)

\bibitem{bendory2022nearoptimal}
T.~Bendory, C.y. Cheng, D.~Edidin, {Near-Optimal Bounds for Signal Recovery
  from Blind Phaseless Periodic Short-Time Fourier Transform}.
\newblock {Journal of Fourier Analysis and Applications} \textbf{29}(1) (2022)

\bibitem{beinert2015ambiguities}
R.~Beinert, G.~Plonka, {Ambiguities in one-dimensional discrete phase retrieval
  from Fourier magnitudes}.
\newblock Journal of Fourier Analysis and Applications \textbf{21}, 1169--1198
  (2015)

\bibitem{thibault2012maximum}
P.~Thibault, M.~Guizar-Sicairos, Maximum-likelihood refinement for coherent
  diffractive imaging.
\newblock New Journal of Physics \textbf{14}(6), 063,004 (2012)

\bibitem{roemer2022wirtinger}
P.~Römer, B.~Diederichs, F.~Filbir, in \emph{The 8th International Conference
  on Computational Harmonic Analysis} (2022)

\bibitem{li2022poisson}
Z.~Li, K.~Lange, J.A. Fessler, Poisson phase retrieval in very low-count
  regimes.
\newblock IEEE Transactions on Computational Imaging \textbf{8}, 838--850
  (2022)

\bibitem{chang2019advanced}
H.~Chang, P.~Enfedaque, J.~Zhang, J.~Reinhardt, B.~Enders, Y.S. Yu, D.~Shapiro,
  C.G. Schroer, T.~Zeng, S.~Marchesini, {Advanced denoising for X-ray
  ptychography}.
\newblock Optics express \textbf{27}(8), 10,395--10,418 (2019)

\bibitem{salditt2020nanoscale}
T.~Salditt, A.~Egner, D.R. Luke, \emph{{Nanoscale Photonic Imaging}} ({Springer
  Nature}, Cham, 2020)

\bibitem{marchesini2013augmented}
S.~Marchesini, A.~Schirotzek, C.~Yang, H.t. Wu, F.~Maia, Augmented projections
  for ptychographic imaging.
\newblock Inverse Problems \textbf{29}(11), 115,009 (2013)

\bibitem{wang2017background}
C.~Wang, Z.~Xu, H.~Liu, Y.~Wang, J.~Wang, R.~Tai, Background noise removal in
  x-ray ptychography.
\newblock Applied optics \textbf{56}(8), 2099--2111 (2017)

\bibitem{bendory2017non}
T.~Bendory, Y.C. Eldar, N.~Boumal, {Non-convex phase retrieval from STFT
  measurements}.
\newblock IEEE Transactions on Information Theory \textbf{64}(1), 467--484
  (2017)

\bibitem{rodenburg2008ptychography}
J.M. Rodenburg, Ptychography and related diffractive imaging methods.
\newblock Advances in imaging and electron physics \textbf{150}, 87--184 (2008)

\bibitem{preskitt2021admissible}
B.~Preskitt, R.~Saab, Admissible measurements and robust algorithms for
  ptychography.
\newblock Journal of Fourier Analysis and Applications \textbf{27}, 1--39
  (2021)

\bibitem{viswanathan2015fast}
A.~Viswanathan, M.~Iwen, in \emph{Wavelets and Sparsity XVI}, vol. 9597 (SPIE,
  2015), pp. 281--288

\bibitem{filbir2021recovery}
F.~Filbir, F.~Krahmer, O.~Melnyk, On recovery guarantees for angular
  synchronization.
\newblock Journal of Fourier Analysis and Applications \textbf{27}(2), 31
  (2021)

\bibitem{forstner2020well}
A.~Forstner, F.~Krahmer, O.~Melnyk, N.~Sissouno, Well-conditioned ptychographic
  imaging via lost subspace completion.
\newblock Inverse Problems \textbf{36}(10), 105,009 (2020)

\end{thebibliography}


\end{document}